\documentclass [twoside, 12pt, intlimits]{amsbook}

\allowdisplaybreaks 

\usepackage[mathscr]{eucal}
\usepackage{amsmath}
\usepackage{amsthm}
\usepackage{amssymb}
\usepackage{amsfonts} 
\usepackage{mathrsfs}
\usepackage{amscd}
\usepackage{enumitem} 
\usepackage{color,wrapfig} 
\usepackage[pdftex]{graphicx}
\usepackage[left=2.5cm,right=2.5cm,top=3.5cm,bottom=2cm]{geometry}
\usepackage{fancyhdr} 
\usepackage{layout}
\usepackage{float} 
\usepackage{txfonts} 
\usepackage{longtable}
\usepackage{mathtools} 
\usepackage[all]{xy}
\usepackage{setspace} 
\usepackage{framed} 
\usepackage{mdframed} 
\usepackage[skins,breakable]{tcolorbox} 
\usepackage[explicit]{titlesec} 
\usepackage{tikz} 
\usetikzlibrary{arrows.meta, matrix}
\usepackage{titletoc} 
\usepackage{hyperref}
\mathtoolsset{showonlyrefs}
\makeindex  

\definecolor{linkblue}{rgb}{0,0,.6}
\definecolor{citered}{rgb}{.7,0,0}
\hypersetup{colorlinks =true, linkcolor=linkblue, citecolor = citered, urlcolor=linkblue}

\renewcommand\thesection{\thechapter.\arabic{section}}

\newtheorem{theorem}[equation]{{\hspace*{-1.5em}}Theorem}
\newtheorem{proposition}[equation]{{\hspace*{-1.5em}}Proposition}
\newtheorem{corollary}[equation]{{\hspace*{-1.5em}}Corollary}
\newtheorem{lemma}[equation]{{\hspace*{-1.5em}}Lemma}
\newtheorem{definition}[equation]{{\hspace*{-1.5em}}Definition}
\newtheorem{remark}[equation]{{\hspace*{-1.5em}}Remark}

\newtheorem{example}[equation]{{\hspace*{-1.5em}}Example}

\newtheorem{definition and proposition}[equation]{Definition and Proposition}
\newtheorem{definition and theorem}[equation]{Definition and Theorem}
\newtheorem{definition and remark}[equation]{Definition and Remark}
\newtheorem{definition and lemma}[equation]{Definition en Lemma}
\newtheorem{definition and example}[equation]{Definition and Example}

{\catcode`\@=11 \uccode`9=`\l \uccode`8=`\o %
 \uppercase{\gdef\striplong@#1#2#3#4\relax{%
  \ifx9#2\ifx8#3\@xp\@xp\@xp\@xp\@xp\@xp\@xp\zap@to@space\fi\fi}}}

\makeatletter
\renewenvironment{proof}[1][\proofname]{\par
  \pushQED{\qed}%
  \normalfont \topsep6\p@\@plus6\p@\relax
  \trivlist
  \item[\hskip\labelsep
        \scshape
    #1\@addpunct{.}]\ignorespaces
}{%
  \popQED\endtrivlist\@endpefalse
}
\makeatother

\numberwithin{equation}{chapter}

\titleformat{\chapter}
  {\Huge\bfseries\selectfont}
  {}{0pt}
  {\tcbox[tcbox width= auto limited,colback=gray!20,colframe=black!100!black]
  {\thechapter.\hskip0.5em#1}}

\titleformat{\section}
  {\Large\bfseries\selectfont}
  {}{0pt}
  {\tcbox[tcbox width= auto limited,colback=gray!20,colframe=white!100!black]{\thesection\hskip0.5em#1}} 
  
\titleformat{\subsection}
  {\bfseries\selectfont}
  {}{0pt}
  {\tcbox[tcbox width= auto limited,colback=gray!20,colframe=white!100!black]{\thesection\hskip0.7em#1}} 

\contentsmargin{2.55em}
\dottedcontents{section}[3.8em]{}{2.3em}{.7pc}
\dottedcontents{subsection}[6.1em]{}{3.2em}{.7pc}

\newcommand{\draft}[1]{} 

\newcommand{\vsp}{\vspace{2mm}}

\def\epsilon{\varepsilon}
\def\phi{\varphi}

\newcommand{\al}{\alpha}
\newcommand{\be}{\beta}
\newcommand{\ga}{\gamma}
\newcommand{\de}{\delta}

\newcommand{\ka}{\kappa}
\newcommand{\lam}{\lambda}

\newcommand{\om}{\omega}

\newcommand{\De}{\Delta}

\newcommand{\uti}{{\ti{u}}}

\newcommand{\Kti}{{\ti{K}}}

\def\C{{\mathbb C}}

\def\H{{\mathbb H}}
\def\N{{\mathbb N}}

\def\R{{\mathbb R}}

\def\T{{\mathbb T}}

\def\Z{{\mathbb Z}}

\newcommand{\mcE}{\mathcal E}
\newcommand{\mcF}{\mathcal F}

\newcommand{\mcI}{\mathcal I}
\newcommand{\mcJ}{\mathcal J}
\newcommand{\mcK}{\mathcal K}

\newcommand{\msD}{\mathscr D}

\newcommand{\ti}{\tilde}
\newcommand{\x}{\times}
\newcommand{\del}{\partial}

\newcommand{\abs}[1]{\left\lvert #1 \right\rvert}

\newcommand{\beq}{\begin{equation}}
\newcommand{\eeq}{\end{equation}}
\newcommand{\beqs}{\begin{equation*}}
\newcommand{\eeqs}{\end{equation*}}

\newcommand{\vungoc}{V\~u Ng\d{o}c}

\DeclareMathOperator{\Id}{Id}

\DeclareMathOperator{\Int}{Int}

\def\slashii#1{\setbox0=\hbox{$#1$}             
\dimen0=\wd0                                 
\setbox1=\hbox{\sl/} \dimen1=\wd1            
\ifdim\dimen0>\dimen1                        
\rlap{\hbox to \dimen0{\hfil\sl/\hfil}}   
#1                                        
\else                                        
\rlap{\hbox to \dimen1{\hfil$#1$\hfil}}   
\hbox{\sl/}                               
\fi}                                         %
\def\slashiii#1{\setbox0=\hbox{$#1$}#1\hskip-\wd0\hbox to\wd0{\hss\sl/\/\hss}}
\newcommand{\dirac}{{\slashii{\partial}}} 

\newcommand{\refdalembert}{Theorem \ref{dalembert} (D'Alembert)}

\newcommand{\refKdVscalMult}{Lemma \ref{KdVscalMult}}

\newcommand{\refKdVtrans}{Lemma \ref{KdVtrans}}

\newcommand{\refKdVmult}{Lemma \ref{KdVmult}}

\newcommand{\refKdVscaling}{Lemma \ref{KdVscaling}}

\newcommand{\refKdVgali}{Lemma \ref{KdVgali}}

\newcommand{\refgronwall}{Lemma \ref{gronwall} (Gronwall)}

\newcommand{\refgenHamFct}{Remark \ref{genHamFct}}

\newcommand{\refdegThree}{Example \ref{degThree}}

\newcommand{\reflaxKdvFirst}{Example \ref{laxKdvFirst}}

\newcommand{\reflaxKdvThird}{Example \ref{laxKdvThird}}

\newcommand{\reflaxKdvAlmost}{Example \ref{laxKdvAlmost}}

\newcommand{\refexampleNoLax}{Example \ref{exampleNoLax}}

\begin{document}

\pagenumbering{arabic}

\begin{titlepage}

\mbox{  }

\vspace{20mm}

\centerline{{\fontsize{1.3cm}{0cm}\selectfont {\bf Selected aspects}}}

\vspace{12mm}

\centerline{{\fontsize{1.3cm}{0cm}\selectfont {\bf of the}}} 

\vspace{12mm}

\centerline{{\fontsize{1.3cm}{0cm}\selectfont {\bf Korteweg-de Vries equation}}}

\vspace{25mm}


\vspace{1mm}

\vspace{30mm}

\centerline{{\LARGE {Sonja Hohloch}} \quad {\Large and} \quad {\LARGE {Federico Zadra}}}

\vspace{4mm}

\centerline{{\Large (University of Antwerp, Belgium)}}

\vspace{4mm}

\centerline{{November 2024}}

\end{titlepage}

\titlepage

\frontmatter

\cleardoublepage

\mbox{ \ }

\vspace{40mm}

\begin{center}

 {\Large {\bf Preface}}

\end{center}

\vspace{20mm}

\begin{quote}
 These lecture notes grew out of notes for courses around Integrable PDEs and the KdV equation given by the authors during the past five years at the University of Antwerp (Belgium). The authors wish to thank their students for all the helpful comments over the past years that helped to improve these lecture notes. Further comments and suggestions are welcome!

 \vspace{10mm}

 \flushright Sonja Hohloch $\&$ Federico Zadra

 \flushright University of Antwerp

 \flushright November 2024

\end{quote}

\cleardoublepage

\tableofcontents

\addcontentsline{toc}{chapter}{Contents}

\mainmatter


\chapter{KdV equation: Motivation and basic properties}

There exist various notions of integrability, depending on the context and its natural properties. 
Very loosely, one may say intuitively that a function, vector field, distribution etc. is `integrable' if it has `sufficiently many nice symmetries'.
From the point of view of analysis there is the main distinction into {\em finite dimensional} and {\em infinite dimensional} integrable systems. 


\vsp

Roughly, {\em finite dimensional} integrable systems live on finite dimensional manifolds. {\em Frobenius integrability} is a natural choice for a differential geometric context when working with vector fields and/or distributions. {\em Liouville integrability}, on the other hand, is tailored for Hamiltonian systems on finite dimensional symplectic manifolds and motivated by their property of energy conservation.

When talking of {\em infinite dimensional} integrable systems, one aims at replacing the concept of systems of integrable ODEs on finite dimensional manifolds by (systems of) PDEs on infinite dimensional function spaces. Even more than in the finite dimensional setting, the notion of integrability depends very much on the given situation. Examples of integrable PDEs are the {\em Korteweg-de Vries equation}, the {\em Nonlinear Schr\"odinger equation}, the {\em Camassa-Holm equation} and others.

\vsp

The discovery of the {\em Korteweg-de Vries equation} was motivated by the following observation: In 1834, J.\ Scott Russell (1808 -- 1882, Scottish civil engineer) noticed a `single wave not changing shape' travelling down a narrow channel in front of a boat, see \cite[page 319]{scottRussell}. Such waves are called {\bf solitons\index{Soliton}}. A nice simulation of such a wave can be found on YouTube, see
\centerline{
\url{https://www.youtube.com/watch?v=D14QuUL8x60}.
}

Although \cite[Footnote on page 360]{boussinesq} already mentioned briefly an equation describing such travelling waves in 1877, the breakthrough came in 1895 when Korteweg and de Vries gave an explicit solution of such a soliton in their work \cite{kortewegDeVries}. A historical overview over the research around the Korteweg-de Vries equation can be found in \cite{deJager}.

\vsp

The literature on integrable systems in infinite dimensions is vast and, when crossing over to physics, unfortunately often lacks mathematical rigor.
A short summary of important facts of infinite dimensional integrable Hamiltonian systems can be found in \cite[Section 5.5]{abrahamMarsden}.
The book \cite[Chapter 12]{vilasi} motivates and presents the most important methods and facts concerning the KdV equation in a concise way.

For the reader's convenience, we list here (in alphabetical order) some books and lecture notes on infinite dimensional integrable systems around the Korteweg-de Vries equation:
\begin{itemize}
 \item [$\circ$]
 \cite{ablowitzClarkson},
 \item [$\circ$]
 \cite{babelonBernardTalon}, 
 \item[$\circ$]
 \cite{batlle}, 
 \item[$\circ$]
 \cite{dickey}, 
 \item [$\circ$]
 \cite{drazinJohnson},
  \item [$\circ$]
 \cite{dunajski},
 \item [$\circ$]
 \cite{eckhausVanHarten},
 \item[$\circ$]
 \cite{kuksin},
 \item[$\circ$]
 \cite{kupershmidt}, 
 \item[$\circ$]
 \cite{miwaJimboDate}, 
 \item [$\circ$]
 \cite{olver}, 
 \item[$\circ$]
 \cite{schneiderUecker},
 \item[$\circ$]
 \cite{vilasi}, 
 \item[$\circ$]
 \cite{zhidkov}. 
\end{itemize}


\clearpage

\section{Revisited: Finite dimensional integrable systems}

\noindent
This section recalls the main definitions and properties of (Liouville) completely integrable Hamiltonian systems in finite dimensions. 

\begin{framed}
 \begin{definition}
  \label{symplectic}
  The pair $(M, \om)$ is a {\bf symplectic manifold\index{Symplectic manifold}} if $M$ is a finite dimensional smooth manifold that admits a closed, nondegenerate $2$-form $\om$. Such an $\om$ is said to be a {\bf symplectic form\index{Symplectic form}}.
 \end{definition}
\end{framed}

Note that symplectic manifolds are always even dimensional, but not all even dimensional manifolds are symplectic. Moreover, locally all symplectic manifolds look the same:

\begin{framed}
 \begin{theorem}[Darboux]
  \label{darboux}
  Locally any symplectic manifold looks like $(\R^{2n}, \sum_{k=1}^n dx_k \wedge dy_k)$ endowed with the coordinates $(x,y)=(x_1, \dots, x_n, y_1, \dots, y_n)$.
 \end{theorem}
\end{framed}

The following class of ODEs lives naturally on symplectic manifolds.

\begin{framed}
 \begin{definition}
  \label{hamODE}
  Let $(M, \om)$ be a symplectic manifold and $H: M \to \R$ a smooth function. The {\bf Hamiltonian vector field\index{Hamiltonian vector field}} $X^H$ of $H$ is defined via $\om( X^H, \cdot) = dH(\cdot)$ 
  and the {\bf Hamiltonian equation\index{Hamiltonian equation}} is given by $X^H(z) = z'$. In local Darboux coordinates $(x, y) \in \R^{2n}$, the Hamiltonian vector field is of the form
  \begin{align*}
   X^H(x,y) & = \bigl( \del_{y_1}H(x,y), \dots , \del_{y_n}H(x, y), - \del_{x_1}H(x, y), \dots, -\del_{x_n}H(x,y) \bigr)^T \\
   & =:\bigl(\del_y H(x, y), - \del_x H(x,y)\bigr)^T
  \end{align*}
and the Hamiltonian equation is given by the system of ODEs
\begin{align*}
 x_i' = \del_{y_i}H(x,y), \quad y_i' = -\del_{x_i}H(x,y) \quad \forall\ 1 \leq i \leq n,
\end{align*}
briefly $x'= \del_y H(x, y)$ and $y'=-\del_xH(x, y)$. The associated flow $\phi^H$ is called {\bf Hamiltonian flow\index{Hamiltonian flow}} of $H$. When working with with a Hamiltonian system, the function inducing the Hamiltonian vector field is often referred to as {\bf Hamiltonian function\index{Hamiltonian function}}. 
  \end{definition}
\end{framed}

Solutions of Hamiltonian systems are {\bf energy conserving\index{Energy conservation}}, i.e., a given solution stays for all times in the same level set of the Hamiltonian. This means in particular that Hamiltonian solutions do not roam throughout the whole manifold but stay for generic Hamiltonians in a subset of dimension $\leq \dim M -1$. If we consider two Hamiltonians such that their Hamiltonian solutions stay in addition in each other's level sets, too, then the solutions stay in the intersection set of the level sets and thus in a subset of dimension $\leq \dim M - 2$. Analytically this is expressed by means of the following notion.

\begin{framed}
 \begin{definition}
  \label{poisson}
  Let $(M, \om)$ be a symplectic manifold. The {\bf Poisson bracket\index{Poisson bracket}} induced by $\om$ of two smooth functions $f, g: M \to \R$ is given by 
  $$\{f, g\}:= \om(X^f, X^g) =-df(X^g)=dg(x^f).$$
  Two smooth functions are {\bf Poisson commutative} if their Poisson bracket vanishes.
 \end{definition}
\end{framed}

Geometrically, Hamiltonian solutions of $f$ stay also within level sets of $g$ and vice versa if and only if $\{f, g\}=0$. This motivates

\begin{framed}
 \begin{definition}
  \label{integrable}
Let $(M, \om)$ be a $2n$-dimensional symplectic manifold. A smooth function $h:=(h_1, \dots, h_n) : M \to \R^n$ is said to be a {\bf (momentum map\index{Momentum map} of a) completely integrable (Hamiltonian) system\index{Completely integrable system}\index{Integrable system}\index{Completely integrable Hamiltonian system}} if
\begin{enumerate}[label=\arabic*)]
 \item 
 $X^{h_1}$, \dots, $X^{h_n}$ are almost everywhere linearly independent.
 \item
 $\{h_i, h_j\}=0$ for all $1 \leq i, j \leq n$ \ {\bf (Poisson commutative)\index{Poisson commutative}}.
\end{enumerate}
The components $h_1, \dots, h_n$ are often referred to as {\bf integrals\index{Integral}} and a completely integrable system is briefly abbreviated by $(M, \om, h)$.
\end{definition}
\end{framed}

Contrary to standard notions in physics, some mathematicians call the momentum map briefly {\em moment map\index{Moment map}}.

\vsp

Examples of completely integrable systems are the spherical pendulum, coupled spin oscillators, coupled angular momenta, the Euler, the Lagrange, and the Kovalevskaya spinning top, etc.

\vsp

For an overview over symplectic geometry and its interaction with integrable systems, we refer the interested reader e.g.\ to the lecture notes and textbooks by \cite{audin-torus}, \cite{audin-integrability}, \cite{bolsinov-fomenko}, \cite{cannas-da-silva}, \cite{hofer-zehnder}, \cite{mcduffSalamon}, and \cite{vungoc book}.


\clearpage

\section{The classical wave equation}

\noindent
As mentioned in the introduction of this chapter, the {\em Korteweg-de Vries equation} describes waves in shallow water. After establishing some facts of the {\em wave equation}, we will use it to motivate the {\em Korteweg-de Vries equation}.

The {\bf 1-dimensional wave equation\index{Wave equation}} is given by
$$
u_{tt} = c^2 u_{xx}
$$
where we have coordinates $(x,t) \in \R \x \R$ and $u: \R \x \R \to \R$, a (constant) {\bf phase velocity\index{Phase velocity}} $c \in \R$ and partial derivatives $u_x:= \del_x u$ and $u_{xx}:=\del_x u_x$ and $u_t:= \del_t u$ and $u_{tt}:= \del_t u_t$.

In 1750, the French mathematician Jean-Baptiste le Rond d'Alembert (1717 -- 1783) solved the 1-dimensional wave equation completely:

\begin{framed}
 \begin{theorem}[{\bf D'Alembert\index{D'Alembert}}]
 \label{dalembert}
 Let $c \in \R$. 
  \begin{enumerate}[label=\arabic*)]
   \item 
    The general solution of the 1-dimensional wave equation\index{Wave equation} $u_{tt} = c^2 u_{xx}$ is of the form 
    $$u(x,t)= F(x-ct) + G(x+ct)$$
    where $F, G: \R \to \R$ are arbitrary smooth functions.
    \item
    Given smooth functions $f, g : \R \to \R$, then the initial value problem
    $$
    \left\{
    \begin{aligned}
     & u_{tt}  =c^2 u_{xx} \qquad  \ \mbox{for } (x,t) \in \R \x \R , \\
     & u(x, 0)  = f(x)  \quad \mbox{for } x \in \R, \\
     & u_t(x, 0)  = g(x)  \quad \mbox{for } x \in \R
    \end{aligned}
    \right.
    $$
   has the unique solution
   $$
   u(x,t)= \frac{1}{2} \bigl( f(x-ct) + f(x+ct) \bigr) + \frac{1}{2c} \int_{x-ct}^{x+ct} g(s) \ ds.
   $$
  \end{enumerate}
 \end{theorem}
\end{framed}

\begin{proof}
 {\bf 1)} We consider the change of coordinates 
 $$
 h: \R^2 \to \R^2 , \quad h(x,t):= \bigl(\xi(x, t), \eta(x, t)\bigr):= (x+ct, x-ct)
 $$
 and set $\uti := u \circ h^{-1}$. We calculate $Dh =  \left( \begin{smallmatrix} \xi_x & \xi_t \\ \eta_x & \eta_t \end{smallmatrix} \right) =\left( \begin{smallmatrix} 1 & c \\ 1 & -c \end{smallmatrix} \right)$ and compute
 \begin{align*}
u_{xx} & = (\uti \circ h)_{xx} = (D\uti.h_x)_x = (\uti _\xi \xi_x + \uti _\eta \eta_x)_x = (\uti _\xi + \uti _\eta)_x  \\ 
& = D(\uti_\xi).h_x + D(\uti_\eta).h_x = \uti_{\xi \xi} \xi_x + \uti_{\xi \eta} \eta_x + \uti_{\eta \xi} \xi_x + \uti_{\eta \eta} \eta_x \\
& =  \uti_{\xi \xi} + 2 \uti_{\xi \eta} + \uti_{\eta \eta}
 \end{align*}
and
\begin{align*}
u_{tt} & =  (\uti \circ h)_{tt} = (D \uti.h_t)_t = (\uti _\xi \xi_t + \uti _\eta \eta_t)_t = (c \uti_\xi - c \uti_\eta)_t =c ( \uti_\xi -  \uti_\eta)_t \\
& = c( D( \uti_\xi).h_t - D(\uti_\eta).h_t) = c( \uti_{\xi \xi} \xi_t + \uti_{\xi_\eta} \eta_t - \uti_{\eta \xi} \xi_t - \uti_{\eta \eta} \eta_t ) \\
& = c^2 (\uti_{\xi \xi} - 2 \uti_{\xi \eta} + \uti_{\eta \eta}).
\end{align*}
Putting this together, we obtain
$$
0 = u_{tt}-c^2 u_{xx} = c^2 (\uti_{\xi \xi} - 2 \uti_{\xi \eta} + \uti_{\eta \eta}) - c^2( \uti_{\xi \xi} + 2 \uti_{\xi \eta} + \uti_{\eta \eta}) = - 4 c^2 \uti_{\xi \eta}.
$$
Thus solving $ u_{tt}= c^2 u_{xx}$ is equivalent to solving $\uti_{\xi \eta}=0$ and transforming back. Integrating $\uti_{\xi \eta}=0$ w.r.t.\ $\xi$ yields $\uti_\eta (\xi, \eta) = H(\eta)$ with an arbitrary smooth function $H: \R \to \R$. Integrating now w.r.t.\ $\eta$ yields 
$$\uti(\xi, \eta) = F(\eta) + G(\xi)$$
for an arbitrary smooth function $G: \R \to \R$ and a function $F: \R \to \R$ with $F' = H$.
Therefore we get
\begin{align*}
u(x,t) & = (\uti \circ h)(x,t) = \uti(\xi(x,t), \eta(x, t)) = F(\eta(x, t)) + G(\xi(x, t)) \\ 
& = F(x-ct) + G(x+ct).
\end{align*}
{\bf 2)} Left to the reader.
\end{proof}

The general solution $u(x,t)=  F(x-ct) + G(x+ct)$ from \refdalembert\ consists, for $c>0$ and $t \to \infty$, of a `wave travelling to the right' with `shape' $x \mapsto F(x-ct)$ and a `wave travelling to the left' with `shape' $x \mapsto G(x+ct)$. 

\begin{framed}
 \begin{example}
Let $c>0$ and let $u(x,t)=  F(x-ct) + G(x+ct)$ be the general solution of the 1-dimensional wave equation $u_{tt}=c^2 u_{xx}$ from \refdalembert.
\begin{enumerate}[label=\arabic*)]
 \item 
For $F(s):= e^{-s^2}$ and $G(s):=0$, the solution $u(x,t)=e^{-(x-ct)^2}$ consists of one wave with the shape of the Gauss curve travelling to the right, i.e., $x \mapsto u(x,0)$ is the standard Gauss curve and $x \mapsto u(x,t_0)$ for $t_0 \in \R$ is the Gauss curve with maximum at $x_{max}=  ct_0$, i.e., translated to the right by $ct_0$.
\item
For $F(s):= e^{-s^2}=:G(s)$, the solution $u(x,t)=e^{-(x-ct)^2} +e^{-(x+ct)^2} $ consists of two `somewhat overlapping' Gauss curve shaped waves. At $t=0$, the waves overlap completely, forming {\em one} big wave of shape $x \mapsto u(x,0)=2e^{-x^2} $. The larger $\abs{t}$ becomes, the more the two waves separate, making their two wave crests more and more distinguishable.
\end{enumerate}
 \end{example}
\end{framed}

More on the wave equation (and some other classical PDEs) may be found in \cite{farlow}.


\clearpage

\section{The KdV equation}

\noindent
In this section, we follow the line of thoughts in \cite{dunajski} to deduce the {\em Korteweg-de Vries equation} from the {\em wave equation}.
Note that the {\em wave equation} describes waves under several simplifying assumption like for example
\begin{itemize}
 \item 
 {\it No dissipation\index{Dissipation}}, i.e., the equation is invariant under `time reversal'\index{Time reversal} $t \mapsto -t$.
 \item
 The amplitude of oscillations is small, i.e., there are no nonlinear terms like for instance $u^2$.
 \item
 {\it No dispersion\index{Dispersion}} ({\bf dispersion\index{Dispersion}} means that waves with different wave length travel with different phase velocity), i.e., the {\bf group velocity\index{Group velocity}} (definition see below) is constant.
\end{itemize}
Now we relax these assumptions by allowing for dispersion and nonlinearity. To this end, consider the two waves $v(x,t) := F(x-ct)$ and $w(x,t):= G(x +ct)$ in the general solution of the 1-dimensional wave equation and compute
\begin{align*}
 v_x(x,t) & = F'(x-ct), \quad v_t(x,t) = -c  F'(x-ct), \\
 w_x(x,t) &= G'(x+ct), \quad w_t(x,t) = c  G'(x+ct).
\end{align*}
Thus the two waves satisfies the PDEs
$$
v_x +\frac{1}{c} v_t =0 \qquad \mbox{and} \qquad w_x -\frac{1}{c} w_t =0.
$$
Now consider the `complex' wave given by $z(x,t):= e^{i(kx - \om(k)t)}$ where $k \in \Z$ and $\om: \R \to \R$. The derivative $\om'$ of $\om$ is called {\bf group velocity\index{Group velocity}}. We calculate 
\begin{align*}
 & z_x(x,t)  =  ik e^{i(kx - \om(k)t)} = ikz(x,t), \\
 & z_{xx}(x,t)  = (ik)^2 e^{i(kx - \om(k)t)} = -k^2 z(x,t) , \\
 & z_t (x,t) = -i\om(k)e^{i(kx - \om(k)t)} = -i \om(k) z(x,t), \\
 & z_{tt} (x,t) = (-i\om(k))^2 e^{i(kx - \om(k)t)} = -\om^2(k) z(x,t).
\end{align*}
Thus $z(x,t)= e^{i(kx - \om(k)t)}$ satisfies the wave equation $z_{tt} = \vartheta^2 z_{xx}$ for $\vartheta=\frac{\om(k)}{k}$.
If we assume $\om(k) = ck$ then $\om'(k) = c$, i.e., the group velocity equals the phase velocity and $\vartheta=c$. Then we obtain again the PDE
$$
z_x -\frac{1}{c} z_t =0.
$$
But if we work for instance with $\om(k):= c(k - \be k^3)$ with $\be \in \R^{\neq 0}$, we will get dispersion, changing the underlying PDEs: We now get
\begin{align*}
& z_t (x,t) = -i\om(k)e^{i(kx - \om(k)t)} = -i c(k - \be k^3) z(x,t)
\end{align*}
and compute in addition
\begin{align*}
& z_{xxx}(x,t) = (ik)^3 e^{i(kx - \om(k)t)} = -ik z(x,t) .
\end{align*}
Therefore we obtain the PDE
$$
z_x + \be z_{xxx} + \frac{1}{c} z_t =0.
$$
Setting $\rho(x,t):= \frac{1}{c} z(x,t)$ and $j(x,t):= z(x,t) + \be z_{xx}(x,t)$, this PDE becomes 
$$
j_x + \rho_t =0.
$$
In this context, $\rho$ is usually called {\bf density\index{Density}} and $j$ is said to be the {\bf current\index{Current}}. In physics, an identity of the form $j_x + \rho_t =0$ is often called {\bf conservation law\index{Conservation law}} or {\bf continuity equation\index{Continuity equation}}. If we modify the current by adding a nonlinear term, i.e., if we consider 
$$
j(x,t):= z(x,t) + \be z_{xx}(x,t) + \frac{\al}{2} z^2 (x,t)
$$
with $\al \in \R^{\neq 0}$, we obtain a PDE of the form
$$
\frac{1}{c} z_t + z_x + \be z_{xxx} + \al z z_x =0
$$
Applying the variable change $x \mapsto x-ct$ and rescaling leads to 
\begin{framed}
$$ u_t -6 u u_x + u_{xxx} =0.$$
\end{framed}
\noindent
This equation is called {\bf Korteweg-de Vries equation\index{Korteweg-de Vries equation}}, short {\bf KdV\index{KdV}}, and is named after the Dutch mathematicians Diederik Johannes Korteweg (1848 -- 1941) and Gustav de Vries (1866 -- 1934) who found in 1895 an explicit solution describing a `soliton wave'.

\begin{framed}
	\begin{remark}
		In the literature, and in these lecture notes as well, the Korteweg–de Vries equation is presented in various forms. All results presented for the Korteweg–de Vries equation apply to a family of equations that vary with the parameter $\kappa$, specifically:
		\begin{equation}
			u_t + \kappa u u_x + u_{xxx} = 0.
		\end{equation}
	\end{remark}
\end{framed}

\begin{framed}
 \begin{theorem}[\cite{kortewegDeVries}]
$$u(x,t):=  -\frac{2 \lam ^2}{\cosh^2 \bigl(\lam (x - 4 \lam^2 t - \tau_0)\bigr)}$$
with parameter $\lam \in \R$ and position shift $\tau_0 \in \R$ is an explicit soliton\index{Soliton} solution of the KdV equation $u_t -6 u u_x + u_{xxx} =0$.
 \end{theorem}
\end{framed}

\begin{proof}
 Calculate the derivatives of $u$ and verify the equation. A constructive proof, i.e., how to come up with this formula, can be found in \cite[Section 2.3.1]{dunajski}.
\end{proof}

We really need {\em all} terms in $u_t -6 u u_x + u_{xxx} =0$ to obtain solutions with soliton behaviour:

\begin{framed}
 \begin{remark}
 \begin{enumerate}[label=\arabic*)]
  \item 
  Solutions of $u_t -6 u u_x =0$ that form a `nice' soliton at time $t=0$ will break (`shock') at some time $t_0>0$ due to discontinuities of the first derivatives.
  \item
  $u_t + u_{xxx} =0$ leads to dispersion, i.e., a `nice' soliton at time $t=0$ will dilute at some time $t_0>0$ into a bigger wave with preceeding and following smaller waves.
 \end{enumerate}
\end{remark}
\end{framed}

At the end of the 1960s, the series of papers authored by \cite{miura}, \cite{miuraGardnerKruskal}, \cite{suGardner}, \cite{gardner}, \cite{kruskalMiuraGardnerZabusky}, and \cite{gardnerGreeneKruskalMiura} found more soliton solutions and conservation laws, in particular solitons of different size and speed either passing through each other from opposite directions or overtaking each other -- without changing the shape except while `overlapping'. There are several videos on YouTube demonstrating this phenomenon, see for example the following two videos:
\centerline{
\url{https://www.youtube.com/watch?v=v5MGNcCnuE4},
}

\centerline{
\url{https://www.youtube.com/watch?v=H4rN3Wr4ctw}.
}


\clearpage

\section{Other PDEs with conservation laws}

\noindent
The KdV equation is not the only PDEs that admit conservation laws and integrability properties. In this section we list some other of these integrable equations.

\begin{framed}
 \begin{example}
  
 The {\bf Sine-Gordon equation\index{Sine-Gordon equation}}
 $$
 u_{tt} - u_{xx} + \sin (u) =0
 $$
 is used when studying surfaces of constant negative curvature. 
 \end{example}
\end{framed}

A YouTube video describing an `artistic' soliton solution of the Sine-Gordon equation can be found at: \url{https://www.youtube.com/watch?v=SAbQ4MvDqEE}.

\begin{framed}
 \begin{example}
  The {\bf nonlinear Schr\"odinger equation\index{Nonlinear Schr\"odinger equation}\index{Schr\"odinger equation}}
 $$
 iu_t = - \frac{1}{2} u_{xx} + \ka \abs{u}^2 u
 $$
 describes propagation of light.  
 \end{example}
\end{framed}

\begin{framed}
 \begin{example}
 The {\bf Camassa-Holm equation\index{Camassa-Holm equation}}
 $$
 u_t + 2 \ka u_x - u_{xxt} + 3uu_x = 2 u_x u_{xx} + u u_{xxx}
 $$
 describes waves in shallow water.  
 \end{example}
\end{framed}

PDEs defined over a discrete space are often called {\bf lattices\index{Lattice}}. An important example was found by the Japanese physicist Morikazu Toda (1917 -- 2010) in 1967:

\begin{framed}
 \begin{example}
 Let $n \in \Z$. Then the {\bf Toda lattice\index{Toda lattice}} is given by
 $$
 \left\{ 
 \begin{aligned}
    q_t(n, t) & = p(n,t), \\
   p_t(n, t) & = e^{-(q(n, t) - q(n-1, t))} - e^{-(q(n+1, t) -q(n,t))}.
 \end{aligned}
 \right. 
 $$
 It describes a discrete chain of particles, parametrized by $n \in \Z$, that interact with their adjacent left and right `neighbours'.
 \end{example}
\end{framed}
 
\noindent
Videos of soliton solutions of the Toda lattice can be found for instance on Gerald Teschl's webpage

\centerline{\url{https://www.mat.univie.ac.at/~gerald/ftp/book-jac/toda.html}}

\vspace{3mm}

\noindent
There are various techniques to study the above mentioned equations, among others
\begin{itemize}
 \item
 Hamiltonian formalism,
 \item
 Integrals of motion as in the finite dimensional case,
 \item
 Hierarchies of equations,
  \item 
 Scattering,
 \item
 Lax pairs and Lax equation.
\end{itemize}
We will encounter them in later sections.


\clearpage

\section{Symmetries of the KdV equation}

\label{sectionSymmetries}

The aim of this section is to describe all transformations that generate new solutions of the KdV equation from a given one. Such transformations are sometimes referred to as {\em symmetries of the equation}\index{Symmetry}.

In this section, $\R^2$ is endowed with the coordinates $(x,t) \in\R^2$ that are mapped under the change of coordinates $h: \R^2 \to \R^2$ to the coordinates $(\xi, \tau)\in \R^2$.

\begin{framed}
 \begin{lemma}
 \label{KdVscalMult}
Let $u \in C^3(\R^2, \R)$ and $\ka \in \R$. Then the transformation $ u \mapsto \ka u$ transforms KdV into
$$
(\ka u)_t + 6 (\ka u) (\ka u_x) + (\ka u)_{xxx} = \ka( u_t + 6 \ka u u_x + u_{xxx}).
$$
 \end{lemma}
\end{framed}

\begin{proof}
 This follows immediately from $(\ka u)_t=\ka u_t$ and $(\ka u)_x=\ka u_x$ and $(\ka u)_{xxx}= \ka u_{xxx}$.
\end{proof}

Thus there is hope that tweaking this simple transformation a bit may lead to a transformation that maps a solution of KdV to a (new) solution of KdV.

\vsp

To this aim, let us briefly recall how derivatives transform under change of coordinates. Let $u\in C^3(\R^2, \R)$ and $h: \R^2 \to \R^2$ a $C^3$-diffeomorphism given by
$$
h(x,t)=(\xi(x,t), \tau (x, t)) =:(\xi, \tau) \quad \mbox{and} \quad h^{-1}(\xi, \tau)=(x(\xi, \tau), t(\xi, \tau))=:(x,t).
$$
Set $\uti:= u \circ h^{-1}: \R^2 \to \R$. When we calculate partial derivatives of this concatenation and want to distinguish between the foot point and a vector we write $ (u \circ h^{-1}) (\xi, \tau) =:u|_{h^{-1} (\xi, \tau)} $ where the latter stands for {\em $u$ at the point $h^{-1} (\xi, \tau)$}. This leads to
\begin{align*}
 (\uti_\xi, \uti_\tau) & = D\uti = D(u \circ h^{-1}) = Du|_{h^{-1}} \circ D(h^{-1}) 
 = (u_x, u_t)|_{h^{-1}} 
 \begin{pmatrix}
  x_\xi & x_\tau \\ t_\xi & t_\tau 
 \end{pmatrix}
\\
&= \begin{pmatrix} u_x|_{h^{-1}} \ x_\xi + u_t|_{h^{-1}} \ t_\xi, & u_x|_{h^{-1}} \ x_\tau + u_t|_{h^{-1}} \ t_\tau \end{pmatrix}
\end{align*}
and, similarly for $u=\uti \circ h$, to
\begin{align*}
 (u_x, u_t) = \begin{pmatrix} \uti_\xi |_h \ \xi_x + \uti_\tau |_h \ \tau_x, & \uti_\xi |_h \ \xi_t + \uti_\tau  |_h \ \tau_t \end{pmatrix}.
\end{align*}
For higher derivatives, we obtain by means of the product and chain rule
\begin{align*}
 \uti_{\xi \xi} & = (\uti_\xi)_\xi = (u_x|_{h^{-1}} \ x_\xi + u_t|_{h^{-1}} \ t_\xi)_\xi \\
 & = (u_x|_{h^{-1}})_\xi \  x_\xi + u_x|_{h^{-1}} \ x_{\xi \xi} + ( u_t|_{h^{-1}})_\xi \ t_\xi +  u_t|_{h^{-1}} \ t_{\xi \xi} \\
 & = (u_{xx}|_{h^{-1}},u_{xt}|_{h^{-1}} ) \begin{pmatrix} x_\xi \\ t_\xi \end{pmatrix}  x_\xi + u_x|_{h^{-1}} \ x_{\xi \xi} + ( u_{tx}|_{h^{-1}},  u_{tt}|_{h^{-1}})  \begin{pmatrix} x_\xi \\ t_\xi \end{pmatrix} t_\xi +  u_t|_{h^{-1}} \ t_{\xi \xi} \\
 & =  u_{xx}|_{h^{-1}} \ x_\xi^2 + u_{xt}|_{h^{-1}} \ x_\xi \ t_\xi+ u_x|_{h^{-1}} \ x_{\xi \xi} + 
  u_{tx}|_{h^{-1}} \ x_\xi \  t_\xi +  u_{tt}|_{h^{-1}} \ t_\xi ^2 +  u_t|_{h^{-1}} \ t_{\xi \xi}  \\
  & =  u_{xx}|_{h^{-1}} \ x_\xi^2 + u_x|_{h^{-1}} \ x_{\xi \xi} + 
  2 u_{xt}|_{h^{-1}} \ x_\xi \  t_\xi +  u_{tt}|_{h^{-1}} \ t_\xi ^2 +  u_t|_{h^{-1}} \ t_{\xi \xi}
\end{align*}
and similarly for the other higher partial derivatives.
When considering $x$ as {\em space variable} and $t$ as {\em time variable}, we get in particular:

\begin{framed}
\begin{lemma}[{\bf Translation invariance in space and time}\index{Translation invariance of KdV}]
\label{KdVtrans}
 Let $u \in C^3(\R^2, \R) $ and $\xi_0, \tau_0 \in \R$ and $h: \R^2 \to \R^2$, $h(x,t)=: (x-x_0, t-t_0)$ and set $\uti := u \circ h^{-1}$, i.e., $\uti(\xi, \tau)=u(\xi+\xi_0, \tau+\tau_0)$. Then $\uti$ is a solution of KdV if and only if $u$ is a solution of KdV.
\end{lemma}
\end{framed}

\begin{proof}
 We calculate
 \begin{align*}
  \uti_\tau + 6 \uti \uti_\xi + \uti_{\xi \xi \xi} & = u_t|_{h^{-1}} + 6 u|_{h^{-1}} u_x|_{h^{-1}} + u_{xxx}|_{h^{-1}} = (u_t+ 6 u u_x + u_{xxx}) \circ h^{-1}
 \end{align*}
 which proves the claim.
\end{proof}

Moreover, we observe

\begin{framed}
\begin{lemma}
\label{KdVmult}
 Let $u \in C^3(\R^2, \R) $ and $a, b \in \R^{\neq 0}$ and $h: \R^2 \to \R^2$, $h(x,t)=: (ax, bt)$ and set $\uti := u \circ h^{-1}$, i.e., $\uti(\xi, \tau)=u\left(\frac{1}{a} \xi, \frac{1}{b} \tau \right)$. Then KdV transforms as follows:
 $$
 \uti_\tau + 6 \uti \uti_\xi + \uti_{\xi \xi \xi} = \left( \frac{1}{b} u_t + \frac{6}{a} u u_x + \frac{1}{a^3} u_{xxx} \right) \circ h^{-1}.
 $$
\end{lemma}
\end{framed}

\begin{proof}
 We calculate $Dh=\left(\begin{smallmatrix} a & 0 \\ 0 & b \end{smallmatrix} \right) $ and $D(h^{-1}) = \left(\begin{smallmatrix} \frac{1}{a} & 0 \\ 0 & \frac{1}{b} \end{smallmatrix} \right)$ and get 
 \begin{align*}
 \uti_\tau = \frac{1}{b}u_t|_{h^{-1}}, \quad \uti_\xi = \frac{1}{a}u_x|_{h^{-1}}, \quad \uti_{\xi \xi} = \frac{1}{a^2}u_x|_{h^{-1}}, \quad  \quad \uti_{\xi \xi \xi} = \frac{1}{a^3}u_x|_{h^{-1}}
 \end{align*}
which implies the claim.
\end{proof}

This leads to

\begin{framed}
\begin{lemma}[{\bf Coupled scaling invariance\index{Coupled scaling invariance of KdV}}]
\label{KdVscaling}
 Let $u \in C^3(\R^2, \R) $ and $a, b, \ka \in \R^{\neq 0}$ and $h: \R^2 \to \R^2$, $h(x,t)=: (ax, bt)$ and set $\uti := \frac{1}{\ka} u \circ h^{-1}$, i.e., $\uti(\xi, \tau) = \frac{1}{\ka} u\left(\frac{1}{a} \xi, \frac{1}{b} \tau \right)$. Then 
 \begin{enumerate}[label=\arabic*)]
  \item 
  $
 \uti_\tau + 6 \uti \uti_\xi + \uti_{\xi \xi \xi} = \left( \frac{1}{b\ka} u_t + \frac{6}{a\ka^2} u u_x + \frac{1}{a^3\ka} u_{xxx} \right) \circ h^{-1}
 $
 \item
 If $\uti$ is a solution of KdV and $\frac{b}{a \ka}=1 = \frac{b}{a^3}$ then $u$ is a solution of KdV.
 \item
 If $u$ is a solution of KdV and $\frac{b}{a \ka}=1 = \frac{b}{a^3}$ then $\uti$ is a solution of KdV.
 \item
  The solutions $(a, b, \ka)$ of $\frac{b}{a \ka}=1 = \frac{b}{a^3}$ are given by $(a, b, \ka)=(\lam, \lam^3, \lam^2)$ for $\lam \in \R$.
 \end{enumerate}
\end{lemma}
\end{framed}

\begin{proof}
{\bf 1)} \refKdVscalMult\ and \refKdVmult\ together yield
\begin{equation}
 \label{eqKdVscaling}
 \uti_\tau + 6 \uti \uti_\xi + \uti_{\xi \xi \xi} = \left( \frac{1}{b\ka} u_t + \frac{6}{a\ka^2} u u_x + \frac{1}{a^3\ka} u_{xxx} \right) \circ h^{-1}
\end{equation}
{\bf 2)} If $\uti$ is a solution of KdV then the left hand side of \eqref{eqKdVscaling} vanishes and multiplication with $b\ka$ leads to 
$$0 =\left( u_t + 6\frac{b}{a\ka} u u_x + \frac{b}{a^3} u_{xxx} \right) \circ h^{-1}$$
Thus, if $\frac{b}{a \ka}=1 = \frac{b}{a^3}$ then $u$ solves KdV. 

\noindent
{\bf 3)} Similar to {\bf 2)}.

\noindent
{\bf 4)} From $ \frac{b}{a^3}=1$ we get $b=a^3$. Using this, $\frac{b}{a \ka}=1 $ leads to $\ka= a^2$.
\end{proof}

Furthermore

\begin{framed}
\begin{lemma}[{\bf Galilean boost\index{Galilean boost}}]
\label{KdVgali}
 Let $h : \R^2 \to \R^2$, $h(x,t)=(x-\lam t, t)$ and let $\lam \in \R$. Then $u \in C^3(\R^2, \R)$ is a solution of KdV if and only if $v:=(u \circ h^{-1}) - \frac{\lam}{6}$ is a solution of KdV.
\end{lemma}
\end{framed}

\begin{proof}
 We calculate $Dh=\left(\begin{smallmatrix} 1 & -\lam \\ 0 & 1 \end{smallmatrix}\right)$. Finding $h^{-1}(\xi, \tau) = (\xi + \lam \tau, \tau)$, we get $D(h^{-1})=\left(\begin{smallmatrix} 1 & \lam \\ 0 & 1 \end{smallmatrix}\right)$ and thus for $\uti: = u \circ h^{-1}$
 \begin{align*}
  \uti_\xi = u_x|_{h^{-1}}, \quad \uti_{\xi \xi} = u_{xx}|_{h^{-1}} , \quad \uti_{\xi \xi \xi} = u_{xxx}|_{h^{-1}} , \quad \uti_\tau = \lam u_x|_{h^{-1}} + u_t|_{h^{-1}}
 \end{align*}
which leads to
\begin{align*}
  \uti_\tau + 6 \uti \uti_\xi + \uti_{\xi \xi \xi} & = (u_t + (\lam + 6u)u_x + u_{xxx} ) \circ h^{-1} \\
   & = \left(u_t + 6\left(\frac{\lam}{6} + u \right)u_x + u_{xxx} \right) \circ h^{-1}
\end{align*}
Now set $v:= \left(u-\frac{\lam}{6} \right) \circ h^{-1}= \uti - \frac{\lam}{6} $ and note that 
$$
v_\tau=\uti_\tau, \quad v_\xi = \uti_\xi, \quad  v_{\xi \xi} = \uti_{\xi \xi}, \quad  v_{\xi \xi \xi} = \uti_{\xi \xi \xi}.
$$
Thus we get
\begin{align*}
 v_\tau + 6 v v_\xi + v_{\xi \xi \xi} & = \uti_\tau + 6 \left(\uti - \frac{\lam}{6} \right) \uti_\xi + \uti_{\xi \xi \xi}  \\
 & = \left( \lam u_x + u_t + 6 \left(u - \frac{\lam}{6} \right) u_x + u_{xxx} \right) \circ h^{-1} \\
 & = (u_t + 6 uu_x + u_{xxx}) \circ h^{-1}
\end{align*}
Thus $v$ is a solution of KdV if and only if $u$ is a solution of KdV. 
\end{proof}

Altogether, we get

\begin{framed}
\begin{corollary}[{\bf Symmetries of KdV}\index{Symmetry}]
 The KdV equation has precisely four `symmetries' that transform solutions into solutions, namely translation in space, translation in time, coupled scaling, and the Galilean boost. 
 \end{corollary}
\end{framed}

\begin{proof}
 \refKdVtrans, \refKdVscaling, and \refKdVgali\ showed that translation in space, translation in time, coupled scaling, and the Galilean boost transform solutions of KdV into new solutions. It is much more involved to prove that this are in fact {\em all} symmetries: We refer to \cite[Chapter 2]{olver} for the necessary theory and to \cite[Example 2.44]{olver} for the application of this theory to the KdV equation where this claim is proven.
\end{proof}


\clearpage

\section{Initial value problem for KdV}

\noindent
When dealing with ODEs, there are the theorems of Peano\index{Peano}{\index{Theorem of Peano} and Picard-Lindel\"of\index{Picard-Lindel\"of}\index{Theorem of Picard-Lindel\"of} that establish existence resp.\ existence $\&$ uniqueness of solutions of initial value problems in a very general way: If the function underlying an ODE is continuous, we have existence of solutions of initial value problems. If the function underlying an ODE is locally lipschitz, we have existence and uniqueness of solutions of initial value problems.

Unfortunately, there is no such general theory of existence and uniqueness for PDEs --- here each class of PDEs has to be studied separately if it admits existence and maybe even uniqueness of solutions of initial value problems. The KdV equation $u_t + 6uu_x + u_{xxx}=0$ indeed admits existence and uniqueness theorems. We will prove a uniqueness theorem using quite strong assumptions to make the proof easier. More general existence theorems with weak(er) assumptions can be found in \cite[Chapter 1.1]{zhidkov} who is dealing in later chapters with other `qualitative' questions of KdV such as stability, convergence, and invariant measures, briefly, with notions related to the theory of classical dynamical systems.

\begin{framed}
\begin{definition}
 Consider $\R^2$ with coordinates $(x, t)$ and let $u: \R \times \R \to \R$ be a function.
 \begin{enumerate}[label=\arabic*), leftmargin=*,align=left]
  \item 
  $u$ is said to have {\bf initial data\index{Initial data} $u_0: \R \to \R$ in $t_0\in \R$} if $u(x, t_0) = u_0(x)$ for all $x \in \R$.
  \item 
  $u$  {\bf decays sufficiently rapidly}\index{Decaying sufficiently rapidly} if $\lim_{x \to \pm \infty} \del_x^k u(x, t) = 0$ for all $t \in \R$ and all $k \in \N_0$.
  \item
  The {\bf energy\index{Energy}} of $u$ is given by 
  $$E_u: \R \to \R^{\geq 0}, \quad E_u(t):= \int_{-\infty}^\infty \abs{u(x,t)}^2 \ dx. $$  
 \end{enumerate}
\end{definition}
\end{framed}

The following statement stems in fact from ODE theory but is often very useful when estimating a growth rate of solutions. It is due to the (originally Swedish) mathematician Thomas Hakon Gr\"onwall (1877 -- 1932) who, after emigrating to the USA, spelt his name {\em Gronwall}. There are in fact two formulations, the first one based on a differential equation, the second one using an integral equation. We will recall (and later use) the first one:

\begin{framed}
 \begin{lemma}[Gronwall\index{Gronwall}\index{Lemma of Gronwall}]
  \label{gronwall}
  Let $t_0$, $T \in \R$ with $t_0<T$ and let $I \subset \R$ be an interval of the form $[t_0, T]$ or $[t_0, T[$ or $[t_0, \infty[$.
   Let $\al, \be \in C^0(I, \R)$ and let $\al$ be in addition differentiable on $\Int(I)$. Moreover, assume that 
   $ \al'(t) \leq \be (t) \al(t)$ for all $ t \in \Int(I)$.
   Then 
   $$\al(t) \leq \al(t_0) \exp\left( \int_{t_0}^t \be(s) \ ds \right) \quad\ \mbox{for all } t \in I. $$
 \end{lemma}
\end{framed}

\begin{proof}
Set $B: I \to \R^{>0}$, $B(t):=\exp\left( \int_{t_0}^t \be(s) \ ds \right)$ and keep in mind that $B(t_0)=\exp(0)=1$ and $B(t)>0$ for all $t \in I$. Note that $B(t)$ satisfies $B'(t) = \be(t) B(t)$ for all $t \in I$ and calculate
\begin{align*}
 \left(\frac{\al}{B}\right)'(t) & = \frac{\al'(t) B(t) - \al(t) B'(t)}{B^2(t)} =\frac{\al'(t) B(t) - \al(t) \be(t) B(t)}{B^2(t)}  \\
 & = \frac{\al'(t) - \al(t) \be(t)}{B(t)} \leq 0
\end{align*}
since $ \al'(t) \leq  \al(t)\be (t)$ by hypothesis and $B>0$. Thus $t \mapsto \frac{\al(t)}{B(t)}$ is a decreasing function and therefore attains it maximum at $t_0$. This means in particular $\al(t_0) = \frac{\al(t_0)}{B(t_0)} \geq \frac{\al(t)}{B(t)}$ for all $t \in I$ which is equivalent to
$$
\al(t) \leq \al(t_0) B(t) = \al(t_0)\exp\left( \int_{t_0}^t \be(s) \ ds \right)
$$
for all $t \in I$, proving the claim.
\end{proof}

One useful application of \refgronwall\ is

\begin{framed}
 \begin{corollary}
 \label{initialNull}
Let $\al$ satisfy the hypotheses of \refgronwall. If $\al(t_0)=0$ and $\al(t)\geq 0$ for all $t>t_0$, then $\al(t)=0$ for all $t\geq t_0$.
 \end{corollary}
\end{framed}

\begin{proof}
 If $\al(t_0)=0$ and $\al(t)\geq 0$ for all $t>t_0$, then \refgronwall\ implies $0 \leq \al(t) \leq 0$, thus $\al(t)=0$ for all $t \geq t_0$.
\end{proof}

Now we are ready to study initial value problems of KdV.

\begin{theorem}
 Smooth solutions of KdV that decay sufficiently rapidly are uniquely determined by their initial data.
\end{theorem}

\begin{proof}
 Let us consider two smooth solutions $u, v : \R \times \R \to \R$ of KdV with initial data $u_0$, $v_0: \R \to \R$ in $t_0 \in \R$ with $u_0 \equiv v_0$ and set $z:=u-v$. This implies in particular $E_z(t_0) = \int_{-\infty}^\infty (u_0-v_0)^2 \ dx =0$. The idea is to employ \refgronwall\ to show $E_z(t)=0$ for all $t>t_0$ and deduce then $u\equiv v$ since $(u-v)^2 \geq 0$.
 
 Now we look for a PDE of which $z$ is a solution. To this aim, plug $z$ into the KdV equation and compute
 \begin{align*}
  & z_t + 6zz_x + z_{xxx} \\
  &  = u_t - v_t + 6(u-v)(u_x-v_x) + u_{xxx} - v_{xxx} \\
  &  = u_t - v_t + 6uu_x -6uv_x -6u_x v + 6vv_x   + u_{xxx} - v_{xxx} -6vv_x + 6vv_x \\
  & = -6uv_x -6u_xv +6vv_x + 6vv_x = -6(u-v)v_x - 6(u_x - v_x) v \\
  & = -6zv_x - 6 z_x v 
 \end{align*}
which yields 
\begin{align*}
 0 & = z_t + 6zz_x + z_{xxx} + 6zv_x + 6 z_x v  =  z_t+ z_{xxx} +6zu_x + 6 z_x v.
\end{align*}
Now multiply the equation by $z$ and obtain
\begin{align*}
 0 & =  zz_t+ zz_{xxx} +6z^2 u_x + 6 zz_x v
\end{align*}
and integrate it from $- \infty $ to $\infty$ w.r.t.\ the variable $x$:
\begin{align*}
 0 & =  \int_{-\infty }^ \infty zz_t \ dx +   \int_{-\infty }^ \infty zz_{xxx} \ dx+   \int_{-\infty }^ \infty 6z^2 u_x \ dx+    \int_{-\infty }^ \infty 6 zz_x v\ dx
\end{align*}
Now analyse each integral separately. Since differentiation w.r.t.\ parameters may be pulled out of the integral if the function is sufficiently regular, we get
\begin{align*}
\int_{-\infty }^ \infty zz_t \ dx  = \frac{1}{2} \int_{-\infty }^ \infty (z^2)_t \ dx =\frac{1}{2} \frac{d}{dt}   \int_{-\infty }^ \infty z^2 \ dx .
\end{align*}
Since sufficiently rapidly decaying functions vanish at $\pm \infty$, their boundary values in partial integration vanish and we obtain
\begin{align*}
  \int_{-\infty }^ \infty zz_{xxx} \ dx = [zz_{xx}]_{-\infty}^\infty - \int_{-\infty }^ \infty z_x z_{xx} \ dx = -   \frac{1}{2}\int_{-\infty }^ \infty \left((z_x)^2 \right)_x \ dx = -   \frac{1}{2} \left[(z_x)^2\right]_{-\infty }^ \infty = 0.
\end{align*}
 Similarly we get
 \begin{align*}
 \int_{-\infty }^ \infty 6 zz_x v\ dx = 3 \left[z^2 v \right] _{-\infty }^ \infty - 3 \int_{-\infty }^ \infty  z^2 v_x \ dx = -6 \int_{-\infty }^ \infty  z^2 \frac{v_x}{2} \ dx.
\end{align*}
 Putting everything together, we get
 \begin{align*}
  0 & = \frac{1}{2} \frac{d}{dt}   \int_{-\infty }^ \infty z^2 \ dx + 0 +   \int_{-\infty }^ \infty 6z^2 u_x \ dx  -6 \int_{-\infty }^ \infty  z^2 \frac{v_x}{2} \ dx \\
  & = \frac{1}{2}  \frac{d}{dt} E_z(t) - 6 \int_{-\infty }^ \infty z^2 \left(\frac{v_x}{2} - u_x \right) \ dx .
 \end{align*}
Since $u$ and $v$ are sufficiently rapidly decaying there exists $ M(t):= \max_{x \in \R} \abs{\frac{v_x(x,t)}{2} - u_x(x,t) } < \infty $ for all $t$ which we use to estimate
\begin{align*}
 \frac{d}{dt} E_z(t) & =  12 \int_{-\infty }^ \infty z^2 \left(\frac{v_x}{2} - u_x \right) \ dx 
 \leq  12 \int_{-\infty }^ \infty z^2 \abs{\frac{v_x}{2} - u_x } \ dx
 \leq 12 \ M(t) \int_{-\infty }^ \infty z^2 \ dx  = 12 \ M(t) \ E_z(t).
\end{align*}
resulting in the differential inequality 
$$ E'_z(t) \leq 12 \ M(t) \ E_z(t) $$
for the energy $E_z$ of $z$. \refgronwall\ leads to the estimate
$$
E_z(t) \leq 12 \ E_z(t_0) \ \exp\left( \int_{t_0}^\infty M(s) \ ds\right) 
$$ 
for all $t_0 \in \R$ and all $t> t_0$. Now recall that, by assumption, $u$ and $v$ have the same initial data and that therefore $E_z(t_0)=0$.  Since $E_z(t)$ is by definition always nonnegative, we get for $t>t_0$
$$0 = E_z(t) = \int_{-\infty }^ \infty (u-v)^2 \ dx .$$
Since $(u-v)^2 $ is smooth, we conclude $(u-v)^2=0$ for all $x \in \R$ and all $t>t_0$ implying $u\equiv v$ for all $x \in \R$ and all $t>t_0$.
\end{proof}


\chapter{KdV equation: Hamiltonian PDE and integrability}

In this chapter, we study how the notions of Hamiltonian vector field, Hamiltonian equation, and integrability extend to the infinite dimensional setting.


\section{Hamiltonian formalism for the KdV equation and the KdV hierarchy}

\noindent
Inspired by \cite[Example 5.5.7]{abrahamMarsden}, we drop mathematical rigor in this section and compute purely formally to get a quick impression of the Hamiltonian formalism in infinite dimension.

Denote by $\mcE$ the space of real-valued functions in one variable that 
\begin{itemize}
 \item 
 admit at least three weak derivatives to render the KdV equation $u_t -6uu_x +u_{xxx}=0$ well-defined in a weak sense;
 \item
 have vanishing boundary terms in the sense that, when integrating by parts, $\int u_x(x) \ v(x) \ dx = \int u(x) \ v_x(x) \ dx$ should always hold true.
\end{itemize}
The latter can be e.g.\ achieved by requiring the functions to have compact support.
For more details, we refer to \cite{batlle}, \cite{dunajski}, and \cite{miwaJimboDate} for the necessary theorie of (pseudo)differential operators and suitable Hilbert and/or Banach spaces. 

Now we equip the function space $\mcE$ with the symplectic form
$$
\om_g(u,v):= \frac{1}{2} \int_{- \infty}^\infty \left( \ \int_{-\infty}^x v(x)u(y) - u(x)v(y) \ dy \right) dx
$$
where $g: \R \to \R$ lies in $\mcE$ and $u,v \in T_g \mcE$, i.e., $x \mapsto u(x) \in T_{g(x)} \R \simeq \R$ and $x \mapsto v(x) \in T_{g(x)} \R \simeq \R$ are vector fields along $x \mapsto g(x)$. Neglecting the foot point $g$, we consider $u$ and $v$ in the following as functions $u,v: \R \to \R$. 

Note that $\om$ is a so-called {\em weak} symplectic form such that we normally should be very careful when switching between the form and Hamiltonian vector fields. Consider the function
\begin{equation}
 \label{hamFctInfinite}
 H: \mcE \to \R, \qquad H(g):= \int_{- \infty}^\infty g^3(x) + \frac{1}{2} (g_x)^2(x) \ dx
\end{equation}
whose Hamiltonian vector field $X^H$ is indirectly defined via
$$
\om(X^H, v) = dH(v).
$$
\begin{framed}
 \begin{lemma}
  The Hamiltonian vector field $X^H$ of $H: (\mcE, \om)\to \R$ defined in \eqref{hamFctInfinite} is given by 
  $$X^H(g)= \del_x (3g^2 - g_{xx})= 6gg_x -g_{xxx}.$$
 \end{lemma}
\end{framed}

\begin{proof}
We calculate first
\begin{align*}
 dH|_g (v) & = \left.\frac{d}{ds}\right|_{s=0} H(g+s v) \\
 & = \left.\frac{d}{ds}\right|_{s=0} \int_{- \infty}^\infty (g+sv)^3(x) + \frac{1}{2} ((g+sv)_x)^2(x) \ dx 
 \end{align*}
This is a parameter depending integral where we --- assuming the integrand to be sufficiently well-behaved --- may switch the order of integration and differentiation:
\begin{align*}
& =  \int_{- \infty}^\infty \left.\frac{d}{ds}\right|_{s=0} \left((g+sv)^3(x) + \frac{1}{2} ((g+sv)_x)^2(x)\right) \ dx \\ 
 & =  \int_{- \infty}^\infty \left. \frac{d}{ds}\right|_{s=0} \left( g^3(x) + 3 g^2(x) s v(x) + 3 g(x) s^2 v^2(x)
  +\frac{1}{2}(g_x)^2(x) + g_x(x) sv_x(x)  + \frac{1}{2}s^2(v_x)^2(x) \right) \ dx \\
 &=  \int_{- \infty}^\infty 3 g^2(x)  v(x) +  g_x(x) v_x(x)\ dx
\end{align*}
and finally conclude by integration by parts
\begin{align*}
  & =  \int_{- \infty}^\infty \left(3 g^2(x) -  g_{xx}(x) \right) v(x) \ dx.
\end{align*}
Now abbreviate $(X^H(g))(z)=:X^H(z) \in T_{g(z)} \R \simeq \R$ and consider
\begin{align*}
 \om_g(X^H(g), v) & = \frac{1}{2} \int_{-\infty} ^\infty \int_{- \infty}^x v(x) X^H(y) - X^H(x)v(y) \ dy \ dx \\
  & = \frac{1}{2} \int_{-\infty} ^\infty v(x) \left(\ \int_{-\infty}^x X^H(y) \ dy \right) - X^H(x)  \left(\ \int_{-\infty}^x v(y) \ dy \right) \ dx
\end{align*}
and obtain by integration by parts
\begin{align*}
 & = \frac{1}{2} \cdot 2 \int_{-\infty} ^\infty v(x) \left( \ \int_{-\infty}^x X^H(y) \ dy \right) \ dx \\
 & =  \int_{-\infty} ^\infty v(x) \left( \ \int_{-\infty}^x X^H(y) \ dy \right) \ dx.
\end{align*}
Comparison of both sides of the equation $\om_g(X^H(g), v) = dH|_g (v)$ yields
$$
\int_{-\infty}^x X^H(y) \ dy  = 3 g^2(x) - g_{xx}(x)
$$
so that we obtain
$$
(X^H(g))(x) =X^H(x) = \del_x (3 g^2(x) -  g_{xx}(x)) = 6g(x)g_x(x) - g_{xxx}(x).
$$
\end{proof}

We conclude

\begin{framed}
 \begin{corollary}
The function $H$ in \eqref{hamFctInfinite} gives rise to the Hamiltonian equation
$$\del_tu = X^H(u) = 6 u u_x -  u_{xxx}$$
where $u: \R \x \R \to \R$. This equation is equivalent to the standard Korteweg-de Vries equation 
$$u_t - 6 uu_x +  u_{xxx}=0.$$
 \end{corollary}
\end{framed}

\begin{proof}
 Since we want to observe the change in time we have to `flow' $g \in \mcE$ by means of a time variable $t\in \R$. Therefore replace $g : \R \to \R$ by $u: \R \x \R \to \R$ seen as $u(x,t)=(u(x))(t)$. We find
 $$\del_tu = X^H(u) = 6 u u_x -  u_{xxx}$$
 which is equivalent to  $u_t - 6 uu_x +  u_{xxx}=0$, i.e., we regain the Korteweg-de Vries equation.
\end{proof}

Moreover

\begin{framed}
 \begin{remark}
For $\rho(x,t):=u(x,t)$ and $j(x,t):= -3 u^2(x,t) + u_{xx}(x,t) $ we obtain the conservation law\index{Conservation law}
$$
\rho_t + j_x=0.
$$ 
 \end{remark}
\end{framed}

More generally

\begin{framed}
 \begin{remark}
 \label{genHamFct}
  Let $f$ be a suitably differentiable function in several variables. The Hamiltonian $H: (\mcE, \om) \to \R$ given by 
  $$H(g):= \int_{-\infty}^\infty f(g(x), g_x(x), g_{xx}(x), \dots)\ dx$$ 
  has as Hamiltonian vector field 
  $$X^H(g)= \del_x \left(\frac{\de f}{\de g} \right)$$
  where~\cite{abrahamMarsden}
  \begin{equation}\frac{\de f}{\de g}:= \del_g f - \del_x \left( \del_{g_x} f \right) +  \del^2_x \left( \del_{g_{xx}} f \right) - \dots
  \end{equation}
 \end{remark}
\end{framed}


\begin{proof}
 Verify the identity 
 $$
 \om_g(X^H(g), v) = \int_{-\infty}^\infty \frac{\de f}{\de g} (x) \ v(x) \ dx = dH|_g(v).
 $$
\end{proof}

Now we want to study integrability notions of the Korteweg-de Vries equation. Define the {\bf Poisson bracket\index{Poisson bracket}} of two functions $K, L: \mcE \to \R $ as
$$\{ K, L\}:= \om(X^K, X^L).$$ 
We will see that the KdV equation is in fact only one item within a whole family of similarly generated partial differential equations. This was discovered and studied in a series of papers by \cite{miura}, \cite{miuraGardnerKruskal}, \cite{suGardner}, \cite{gardner}, \cite{kruskalMiuraGardnerZabusky}, \cite{gardnerGreeneKruskalMiura} resulting in:

\begin{framed}
 \begin{theorem}
  For $g \in \mcE$, set $X_1(g):= g_x$ and $f_1(g):= \frac{1}{2}g^2$ and let $a, b \in \R$. Abbreviate $\msD:= \del_x$ which has (on a suitable Hilbert space) an inverse integration operator denoted by $\msD^{-1}$. For $j \geq 2$, define
  \begin{align*}
   X_j(g) & := \bigl(a g + a \msD g \msD^{-1} + b \msD^2 \bigr) X_{j-1} \\
   & \ = \bigl(a g \msD + a \msD g + b \msD^3 \bigr) \frac{\de f_{j-1}}{\de g},
  \end{align*}
i.e., we have $X_j(g) = \del_x \left( \frac{\de f_{j}}{\de g} \right).$ 
\begin{enumerate}[label=\arabic*)]
 \item 
Then the family of equations
$$
u_t = X_j(u), \quad \mbox{for } j\geq 1
$$
is called {\bf KdV hierarchy\index{KdV hierarchy}\index{Hierarchy}} or {\bf higher order KdV equations\index{Higher order KdV equations}\index{KdV equations of higher order}} and recovers for $a=2$, $b=-1$ and $j=2$ the standard KdV equation $u_t-6uu_x + u_{xxx} = 0$.
\item
The higher order KdV equations $u_t = X_j(u)$ are Hamiltonian with Hamiltonian functions $F_j(g):= \int_{- \infty}^\infty f_j(g(x)) \ dx$ and Hamiltonian vector fields $X^{F_j} = X_j$.
\item
The higher order KdV equations are integrable in the sense that $\{ F_j, F_k\}=0$ for all $j,k \geq 1$. 
\item
The Hamiltonian $H$ from the standard KdV equation satisfies $H=F_2$ and therefore $\{H, F_k\}=0$ holds true for all $k\geq 1$. Thus all $F_k$ are integrals of the standard KdV equation. 
\end{enumerate}
 \end{theorem}
\end{framed}

\begin{proof}
 {\bf 1)} 
 $X_1(u) := u_x$ leads to $u_t = X_1(u)=u_x$. Now keep in mind that $u_x = \msD u$ and calculate with $a=2$ and $b=-1$
 \begin{align*}
  u_t & = X_2(u) = (2u + 2 \msD u \msD^{-1} - \msD^2) \msD u \\
  & = 2u \msD u + 2 \msD(u \msD^{-1} \msD u) - \msD^3 u \\
  & = 2uu_x + 2 \del_x(u^2) - u_{xxx} \\
  & = 2uu_x + 4uu_x -u_{xxx} \\
  & = 6uu_x - u_{xxx}
 \end{align*}
which is the standard KdV equation $u_t -6uu_x + u_{xxx}=0$.

{\bf 2)}
Using the identity $X_j(g) = \del_x \left( \frac{\de f_{j}}{\de g} \right)$ and applying \refgenHamFct\ to the Hamiltonian $F_j(g)= \int_{- \infty}^\infty f_j(g(x)) \ dx$, we obtain $X^{F_j}= X_j$. 

{\bf 3)}
In order to show integrability we first prove a recurrence relation for the Poisson bracket. We use the short notation $X_j(z):=X_j(g)(z)$ and compute
\begin{align*}
 \{F_j, F_k\} & = \om(X^{F_j}, X^{F_k}) = \om(X_j, X_k) \\
 & = \frac{1}{2} \int_{- \infty}^ \infty \int_{-\infty}^x X_k(x) X_j(y) - X_j(x) X_k(y) \ dy \ dx \\
 & = \frac{1}{2} \int_{- \infty}^ \infty X_k(x) \left( \ \int_{-\infty}^x  X_j(y) \ dy \right) - X_j(x)\left( \ \int_{-\infty}^x X_k(y) \ dy \right) \ dx \\
 & = \frac{1}{2} \int_{- \infty}^ \infty X_k(x) \left( \ \int_{-\infty}^x   \del_y \left( \frac{\de f_{j}}{\de g}(y) \right)  \ dy \right) - X_j(x)\left( \ \int_{-\infty}^x \del_y \left( \frac{\de f_{k}}{\de g} (y)\right) \ dy \right) \ dx \\
 & = \frac{1}{2} \int_{- \infty}^ \infty X_k(x) \left( \frac{\de f_{j}}{\de g} (x)\right)  - X_j(x) \left( \frac{\de f_{k}}{\de g} (x)\right) \ dx .\\
\end{align*}
Integration by parts leads to
\begin{align}
\label{backwards}
 & =  \int_{- \infty}^ \infty X_k(x) \left( \frac{\de f_{j}}{\de g} (x)\right) \ dx \\ \notag
 & = \int_{-\infty}^\infty  \left( \bigl(a g \msD + a \msD g  + b \msD^3 \bigr) \frac{\de f_{k-1}}{\de g} \right)(x) \left( \frac{\de f_{j}}{\de g} (x)\right)  \  dx
\end{align}
and another integration by parts yields
\begin{align}
 \notag
 & = - \int_{-\infty}^\infty \left( \frac{\de f_{k-1}}{\de g} (x)\right) \left( \bigl(a  \msD g + a g \msD   + b \msD^3 \bigr) \frac{\de f_{j}}{\de g} \right)(x) \ dx \\ 
 \label{fromHere}
 & = - \int_{-\infty}^\infty \left( \frac{\de f_{k-1}}{\de g} (x) \right) X_{j+1}(x) \ dx
\end{align}
Reversing the steps that lead to \eqref{backwards}, we can turn \eqref{fromHere} into
\begin{align*}
 & = - \om(X_{k-1}, X_{j+1}) = \om(X_{j+1}, X_{k-1}) \\
 & = \{F_{j+1}, F_{k-1}\}
\end{align*}
thus obtaining the relation
$$\{F_j, F_k\}= \{F_{j+1}, F_{k-1}\}.$$
Now we show that in fact $\{ F_j, F_k\} = 0$ holds true for all $j,k \geq 1$.
First, consider the case $\abs{j-k}$ even and assume w.l.o.g.\ that $j<k$. By increasing $j \mapsto j+1$ and decreasing $k \mapsto k-1$ simultaneously $\frac{\abs{j-k}}{2}$ times, we obtain
\begin{align*}
 \{F_j, F_k\} & = \{F_{j+1}, F_{k-1} \} = \dots = \left\{ F_{\frac{j+k}{2}}, F_{\frac{j+k}{2}} \right\} =0.
\end{align*}
Second, consider the case $\abs{j-k}$ odd and assume w.l.o.g.\ that $j<k$. By increasing $j \mapsto j+1$ and decreasing $k \mapsto k-1$ simultaneously $\abs{j-k}$ times, we obtain
\begin{align*}
 \{F_j, F_k\} & = \{F_{j+1}, F_{k-1} \} = \dots = \{F_{k-1}, F_{j+1}\} = \{F_k, F_j\}.
\end{align*}
Since the Poisson bracket is skew symmetric, we also have 
$$ \{F_j, F_k\} = - \{F_k, F_j\}.$$ 
Both identities together imply $\{F_j, F_k\} =0$.

{\bf 4)} 
This follows immediately from {\bf 3)}.
\end{proof}


\clearpage

\section{Evolution equations and their symmetries}

\noindent
Intuitively, an integrable system is a system of differential equations whose behavior is determined by initial conditions and which can be solved (`integrated') from those initial conditions. So it is of interest to study the interaction of initial conditions and symmetries.

In Section \ref{sectionSymmetries}, we determined four symmetries of the KdV equation by direct computation. Whereas symmetries of differential equation are often studied by means of Lie groups like e.g.\ in \cite{olver}, we consider symmetries in this section as `variations of the equation that do not destroy solutions'. To this aim, let us explain what we mean with `variation'. For $k,n \in \N_0$, consider
$$
{\small
\mcF_k(\R^{n}, \R):=
\left \{ 
f: \R^n \times \R^k \to \R 
\left| \ f \mbox{ depends on } 
\begin{aligned} 
 & \bullet \mbox{ variables } x=(x_1, \dots , x_n) \in \R^n , \\
 & \bullet \mbox{ parameters } t =(t_1, \dots, t_k) \in \R^k
\end{aligned}
\right. 
\right\}
}
$$
and set
$$
\mcF_{par}(\R^{n}, \R):= \bigcup_{k\geq 0} \mcF_k(\R^{n}, \R).
$$
The idea is to consider a function ${\mathbf f} \in \mcF_0(\R^n, \R)$ and to see it as part of a family $t \mapsto f( \cdot, t) $ with $ {\mathbf f} =f( \cdot, t_0)$ for a suitable parameter value $t=t_0$ and function $f \in \mcF_1(\R^{n}, \R)$.
Given an operator $K$ defined on suitable subspaces of $\mcF_{par}(\R^{n}, \R)$, a special case are families arising from solutions of initial value problems like for instance
$$
\del_t f =K(f) \quad \mbox{with} \quad f(x,t_0) ={\mathbf f} (x) \quad \forall\ x \in \R^n.
$$
This initial value problem describes the `evolution' of ${\mathbf f}$ under the operator $K$ which motivates

\begin{framed}
 \begin{definition}
  Let $K$ be an operator defined between suitable subspaces of $\mcF_{par}(\R^{n}, \R)$. Then $\del_t f =K(f)$ is called {\bf evolution equation\index{Evolution equation}} of the operator $K$.
 \end{definition}
\end{framed}

If $t=(t_1, \dots, t_k)$ is $k$-dimensional, then $\del_t f =K(f)$ stands for the system
\begin{equation*}
\left\{ \
\begin{aligned}
 \del_{t_1} f & = K(f), \\
 \vdots \ \  & \qquad \vdots \\
 \del_{t_k} f & = K(f).
\end{aligned}
\right. 
\end{equation*}

\begin{framed}
 \begin{example}
Let $K: C^\infty(\R \times \R, \R) \to C^\infty(\R \times \R, \R)$ with $K(u):=-(6uu_x + u_{xxx})$ and $u$ depending on $(x,t) \in \R^2$. Then the associated evolution equation $\del_t u = K(u)$ is equivalent to the KdV equation 
$$u_t +6uu_x + u_{xxx}=0.$$
 \end{example}
\end{framed}

Intuitively, a `symmetry' of an evolution equation is another evolution equation that preserves the property of being solutions of the first one.

\begin{framed}
 \begin{definition}
 Let $K$ and $\Kti$ be operators between suitable subspaces of $\mcF_{par}(\R^n, \R)$, let $k, \ell \in \N$, and let the parameters $t=(t_1, \dots, t_k)$ and $s=(s_1, \dots, s_\ell)$ be associated to the evolution equations $\del_t u = K(u)$ and $\del_s u= \Kti(u)$. 
 The evolution equation $\del_s u = \Kti(u)$ is a {\bf symmetry\index{Symmetry}} of the evolution equation $\del_t u = K(u)$ if, given $(x,t) \mapsto {\mathbf u}(x,t)$ satisfying $\del_t {\bf u} = K({\mathbf u})$ and extending to $(x,t,s) \mapsto u(x,t,s)$ with $\del_s u = \Kti(u)$ for all $t \in \R^k$, the function $(x,t) \mapsto u(x,t,s)$ is a solution of $\del_t u = K(u)$ for all $s \in \R^\ell$.
 Briefly, we say that $\Kti $ {\bf gives rise to a symmetry\index{Symmetry} of} $K$.
  \end{definition}
\end{framed}

In other words, let $(x,t) \mapsto {\mathbf u}$ satisfy $\del_t {\mathbf  u} = K({\mathbf u})$. If ${\mathbf u}$ can be considered as the initial condition at some $s=s_0\in \R^\ell$ of the initial value problem
  $$
  \del_s u = \Kti(u), \qquad u(x,t,s_0)={\mathbf u}(x,t) \quad \forall\ (x,t) \in \R^n \times \R^k
  $$
and if all solutions $(x,t,s) \mapsto u(x,t,s)$ of this initial value problem satisfy $\del_t u(x,t,s) = K(u(x,t,s))$ for all $s \in \R^\ell$, then $\del_s u = \Kti(u)$ is a symmetry of $\del_t u = K(u)$.

\vsp

Note that the definition of {\em symmetry} is symmetric in $s$ and $t$ in the following sense: start with a function $x \mapsto u(x,t_0, s_0)$ and vary it first w.r.t.\ $t$ and then w.r.t.\ $s$ and vice versa, i.e., applying $\del_s \del_t$ resp.\ $\del_t \del_s $ to $u$ as sketched in the following diagram
\begin{equation*}
 \begin{array}{ccccccc}
&&& \del_s u = \Kti(u) &&& \\
 & u(\cdot, t, s_0) & \quad   & \longrightarrow & \quad &  u(\cdot, t, s) \\
 &&&&&& \\
\del_t u = K(u)  & \uparrow &&&&  \uparrow & \del_t u = K(u)\\
 &&&&&& \\
 & u(\cdot, t_0, s_0) & \quad   & \longrightarrow & \quad &  u(\cdot, t_0, s) \\ 
 &&& \del_s u = \Kti(u) &&& 
 \end{array}
\end{equation*}
then the diagram commutes if $u$ satisfies $u_{st}=u_{ts}$.

\begin{framed}
 \begin{corollary}[Criterium]
 Let $K$ and $\Kti$ be operaters of which the solutions of their evolution equations $\del_t u= K(u)$ and $\del_s u = \Kti(u)$ all satisfy $\del_s \del_t = \del_t \del_s $. Then $\del_s u = \Kti(u)$ is a symmetry of $\del_t u = K(u)$ (and vice versa) if and only if $ \del_s K(u) = \del_t \Kti(u)$.
 \end{corollary}
\end{framed}

\begin{proof}
We compute $\del_s K(u) = \del_s (\del_t u) = \del_t (\del_s u) = \del_t \Kti(u)$.
\end{proof}

Now let us consider a few examples.

\begin{framed}
 \begin{example}
  Consider the operator $K(u):= -(6uu_x + u_{xxx})$ of which the evolution equation $\del_t u = K(u)$ is equivalent to the KdV equation $u_t + 6uu_x + u_{xxx}=0$. Then $\del_s u = \Kti(u)$ with $\Kti(u):= u^2$ is no symmetry of $\del_t u = K(u)$.
 \end{example}
\end{framed}

\begin{proof}
 We show that $\del_t \Kti(u)= \del_s K(u)$ does not hold: keep in mind that $\del_s \del_x = \del_x \del_s$ holds true for smooth functions and compute
 \begin{align*}
  \del_s K(u) & = -(6uu_x + u_{xxx})_s = -6u_s u_x -6 u(u_x)_s - (u_{xxx})_s \\
   & = -6u_s u_x -6 u (u_s)_x -(u_s)_{xxx} 
 \end{align*}
and now make use of $\del_s u = \Kti(u)=u^2$ and obtain
\begin{align*}
 &= -6u^2 u_x - 12 u^2u_x -2(uu_x)_{xx} = -18u^2 u_x  -2( (u_x)^2 + uu_{xx})_x \\
 & = -18u^2 u_x -4u_x u_{xx} - 2u_xu_{xx} - 2 u u_{xxx} =  -18u^2 u_x -6u_x u_{xx}- 2 u u_{xxx}.
\end{align*}
On the other hand, we have $u_t= K(u) = -6uu_x -u_{xxx}$ and thus we get
\begin{align*}
 \del_t \Kti(u) & = \left(u^2\right)_t = 2 uu_t = 2u (-6uu_x -u_{xxx}) = -12u^2u_x -2uu_{xxx}
\end{align*}
so that $\del_t \Kti(u) \neq \del_s K(u)$, i.e., $\Kti$ does not give rise to a symmetry for $K$.
\end{proof}

We will now consider examples based on `homogeneous differential polynomials'. To this aim, we assign degrees to differential terms as follows: a term of differential order zero is assigned degree 2. Each differentiation w.r.t.\ $x$ adds $+1$ to the degree, i.e., $\deg(u)=2$, $\deg(u_x)=3$, $\deg(u_{xx})=4$ etc., and multiplication leads to addition of the degrees, i.e., $\deg(uu_x)=2+3=5$ etc. This gives the following monomials of odd degree:
\begin{itemize}[leftmargin=10mm]
 \item 
 degree 3: \ $u_x$
 \item
 degree 5: \ $uu_x$, $u_{xxx}$
 \item
 degree 7: \ $u^2u_x$, $uu_{xxx}$, $u_x u_{xx}$, $u_{xxxxx}$
\end{itemize}
This leads to the following homogeneous polynomials:
\begin{itemize}[leftmargin=10mm]
 \item 
 degree 3: \ $cu_x$ \ with (constant) coefficient $c \in \R$.
 \item
 degree 5: \ $c_1 uu_x + c_2u_{xxx}$ \ with (constant) coefficients $c_1, c_2 \in \R$.
 \item
 degree 7: \ $c_1 u^2u_x + c_2 uu_{xxx} + c_3 u_x u_{xx} + c_4 u_{xxxxx}$ \ with (constant) coefficients $c_1, c_2, c_3, c_4 \in \R$.
\end{itemize}
For homogeneous polynomials of degree 3, we obtain:

\begin{framed}
 \begin{example}
 \label{degThree}
  Consider the operator $K(u):= -(6uu_x + u_{xxx})$ of which the evolution equation $\del_t u = K(u)$ is equivalent to the KdV equation $u_t + 6uu_x + u_{xxx}=0$. Then $\del_s u = \Kti(u)$ with $\Kti(u):= c u_x$ is a symmetry of $\del_t u = K(u)$ for all $c \in \R$.
 \end{example}
\end{framed}

\begin{proof}
 We keep in mind that $\del_s \del_x = \del_x \del_s$ holds true for smooth functions and that, by assumption, $u_s= cu_x$. Then we get
 \begin{align*}
  \del_s K(u) & = -(6uu_x + u_{xxx})_s = -6u_su_x -6u(u_s)_x - (u_s)_{xxx} \\
  & = -6c(u_x)^2 - 6cuu_{xx} - cu_{xxxx}.
\end{align*}
On the other hand, keeping $u_t = -(6uu_x + u_{xxx})$ in mind, we compute
\begin{align*}
 \del_t \Kti(u) = (c u_x)_t = c(u_t)_x =  -6c uu_{xx} - 6c(u_x)^2  - cu_{xxxx}
\end{align*}
which yields $ \del_s K(u) = \del_t \Kti(u)$, i.e., $\Kti$ gives rise to a symmetry.
\end{proof}

\refdegThree\ corresponds to the KdV symmetry {\em translation in space}. This makes intuitively sense since the variation with $s$ is in fact a `variation in space' via $u_s=u_x$ (see also \cite[Remark 1.1]{miwaJimboDate}).
For homogeneous polynomials of degree 5, we obtain:

\begin{framed}
 \begin{example}
 \label{degFive}
  Consider the operator $K(u):= -(6uu_x + u_{xxx})$ of which the evolution equation $\del_t u = K(u)$ is equivalent to the KdV equation $u_t + 6uu_x + u_{xxx}=0$. Then $\del_s u = \Kti(u)$ with $\Kti(u)= 6c uu_x + c u_{xxx}$ is a symmetry of $\del_t u = K(u)$ for all $c\in \R$ and it is the only one arising from the ansatz $\Kti(u) = c_1 uu_x + c_2 u_{xxx}$ with constants $c_1, c_2 \in \R$.
 \end{example}
\end{framed}

\begin{proof}
 We make the ansatz $\Kti(u):= c_1 uu_x + c_2 u_{xxx}$ and will show that $6c_2=c_1$ must hold true in order to give rise to a symmetry. We keep in mind that $\del_s \del_x = \del_x \del_s$ holds true for smooth functions and that by assumption $u_s= c_1 uu_x + c_2 u_{xxx}$. Then we get
 \begin{align*}
  \del_s K(u) & = -(6uu_x + u_{xxx})_s = -6u_su_x -6u(u_s)_x - (u_s)_{xxx} \\
  & = -6(c_1 uu_x + c_2 u_{xxx})u_x  -6u (c_1 uu_x + c_2 u_{xxx})_x - (c_1 uu_x + c_2 u_{xxx})_{xxx} \\
  & = -6c_1 u (u_x)^2 - 6c_2 u_x u_{xxx} - 6c_1 u^2 u_{xx} - 6c_1 u(u_x)^2 -6c_2u u_{xxxx}   -c_1 (uu_x)_{xxx} -c_2 u_{xxxxxx} \\
  & = -12c_1 u (u_x)^2 - 6c_2 u_x u_{xxx} - 6c_1 u^2 u_{xx}  -6c_2u u_{xxxx}  -c_1 (uu_x)_{xxx} -c_2 u_{xxxxxx} .
 \end{align*}
On the other hand, making use of $u_t=K(u)=-(6uu_x + u_{xxx})$, we compute
\begin{align*}
 \del_t \Kti(u) & = ( c_1 uu_x + c_2 u_{xxx})_t = c_1 u_t u_x + c_1 u (u_t)_x + c_2 (u_t)_{xxx} \\
 & = -6c_1 u (u_x)^2- c_1 u_x u_{xxx} - 6c_1 u^2 u_{xx} - 6c_1 u (u_x)^2 -c_1 uu_{xxxx}  -6 c_2 (u u_x)_{xxx} - c_2 u_{xxxxxx} \\
 & = -12c_1 u (u_x)^2- c_1 u_x u_{xxx} - 6c_1 u^2 u_{xx} -c_1 uu_{xxxx}  -6 c_2 (u u_x)_{xxx} - c_2 u_{xxxxxx}.
\end{align*}
Thus, the equation $\del_s K(u) = \del_t \Kti(u)$ holds true if and only if $c_1 = 6c_2$. This means that $\Kti$ must be of the form $\Kti(u)= 6c uu_x + c u_{xxx}$ with $c \in \R$.
\end{proof}

\refdegThree\ corresponds to the KdV symmetry {\em translation in time}. This makes intuitively sense since the variation with $s$ is in fact a `variation in time' via $u_s= -(6uu_x + u_{xxx}) = u_t$ (see also \cite[Remark 1.1]{miwaJimboDate}).
For homogeneous polynomials of degree 7, one finds

\begin{framed}
 \begin{example}
 Consider the operator $K(u):= -(6uu_x + u_{xxx})$ of which the evolution equation $\del_t u = K(u)$ is equivalent to the KdV equation $u_t + 6uu_x + u_{xxx}=0$. Then $\del_s u = \Kti(u)$ with $\Kti(u):= c u^2u_x + 2c uu_{xxx} + 4c u_x u_{xx} + \frac{6}{5}c u_{xxxxx}$ is a symmetry of $\del_t u = K(u)$ for all $c\in \R$ and it is the only one arising from the ansatz $\Kti(u) = c_1 u^2u_x + c_2 uu_{xxx} + c_3 u_x u_{xx} + c_4 u_{xxxxx} $ with constants $c_1, c_2, c_3, c_4 \in \R$.
 \end{example}
\end{framed}

\begin{proof}
 Left to the reader; see also \cite[Equation 1.20 and Equation 1.21]{miwaJimboDate}.
\end{proof}


\clearpage

\section{Lax formalism of an evolution equation}

\noindent
Recall that, intuitively, an integrable system is a system of differential equations whose behavior is determined by initial conditions and which can be solved (`integrated') from those initial conditions. So it is of particular interest to find equivalent formulations that are hopefully easier to solve.

Let $L$ and $M$ be operators between suitable subspaces of $ \mcF_{par}(\R^n, \R)$. The {\bf eigenvalue problem\index{Eigenvalue problem}} of $L$ is given by the equation
$$
Lv = \lam v
$$
where $v$ is an eigenfunction of the eigenvalue $\lam$ and both are thought to depend on the same parameter $t \in \R^k$ (if any). Recall that the evolution equation of $M$ is of the form $\del_t u = Mu$. The idea is now to consider the combined equations
\begin{equation}
 \label{laxSystem}
 \left\{ 
    \begin{aligned}
     Lv & = \lam v, \\
     \del_t v & = M v
    \end{aligned}
 \right.
\end{equation}
in the sense of varying the eigenvalue and associated eigenfunction of $L$ with the evolution equation of $M$, i.e., for a given ${x \mapsto \mathbf v}(x)$ with $L {\mathbf v} = \lam {\mathbf v}$ that can be seen as initial condition ${\mathbf v}= v( \cdot, t_0)$ with $\lam=\lam(t_0)$ at time $t=t_0$ consider the solutions $(x, t) \mapsto v(x,t)$ with eigenvalues $t \mapsto \lam(t)$ of the initial value problem
\begin{equation*}
 \left\{ 
    \begin{aligned}
    & L v = \lam v  \quad (\mbox{where } \lam: \R^k \to \R),\\
     & \del_t v  = M v, \\
     & v(x,t_0)= {\mathbf v}(x,t_0) \quad \forall\ x \in \R^n.
    \end{aligned}
 \right.
\end{equation*}
We now rewrite the system in Equation \eqref{laxSystem} as follows: first, differentiate the equation $Lv=\lam v$ w.r.t.\ the parameter $t$. This yields, using short notation $L_t:= \del_tL$ and $v_t:= \del_t v$ and $\lam_t:=\del_t \lam $, the equation
\begin{equation*}
 L_t v + L v_t = \lam_t v + \lam v_t.
\end{equation*}
Making use of $v_t = Mv$, we get
\begin{equation*}
 L_t v + L Mv = \lam_t v + \lam Mv =  \lam_t v + M (\lam v) =  \lam_t v + M Lv.
\end{equation*}
Using the notion of {\bf commutator\index{Commutator}} or {\bf Lie bracket\index{Lie bracket}} of operators 
$$[L,M]:= LM-ML$$ 
we obtain the equation 
\begin{equation*}
 (L_t + [L,M])v = \lam_t v.
\end{equation*}

\begin{framed}
 \begin{definition}
  The equation $L_t + [L,M]=0$ is called {\bf Lax equation\index{Lax equation}}. Two operators $L$ and $M$ satisfying $L_t + [L,M]=0$ are called a {\bf Lax pair\index{Lax pair}}.
 \end{definition}
\end{framed}

This equation is named after the Hungarian-American mathematican Peter David Lax (born in 1926). An immediate observation is

\begin{framed}
 \begin{lemma}
  The Lax equation is satisfied for a pair of operators $L$ and $M$ if and only if the spectrum of $L$ does not change under time evolution, i.e., the eigenvalues $\lam$ do not depend on $t$.
 \end{lemma}
\end{framed}

\begin{proof}
The equation $(L_t + [L,M]) = \lam_t$ implies 
\begin{equation*}
 L_t + [L,M]=0 \qquad \Leftrightarrow \qquad 0=\lam_t = \del_t \lam
\end{equation*}
i.e., the eigenvalues $\lam$ are constant.
\end{proof}

\begin{framed}
	\begin{example}
		Given the Harmonic oscillator $H(q,p) = \frac{p^2}{2 m} + \frac{k}{2} q^2$, its Hamiltonian equations
		\begin{equation}
			 \begin{cases}
				\dot{q} = \frac{p}{m}, \\
				\dot{p} = - k q
			\end{cases}
			\quad \mbox{can be written in matrix form} \quad
			\begin{pmatrix}
				\dot{q} \\ \dot{p}
			\end{pmatrix} = 
			\begin{pmatrix}
				0 & \frac{1}{m} \\ - k & 0
			\end{pmatrix}
			\begin{pmatrix}
				q \\ p
			\end{pmatrix}.
		\end{equation}
		Using the change of variables
		\begin{equation}
			(q,p) \mapsto (Q,P) = \left(\sqrt{\frac{k}{m}} q, \frac{p}{m}\right)
		\end{equation}
		we rewrite the linear system as
		\begin{equation}
			\frac{d}{dt}\begin{pmatrix}
				Q \\ P
			\end{pmatrix}
			=
			\begin{pmatrix}
					0 & \frac{k}{m} \\
					- \frac{k}{m} & 0
			\end{pmatrix}
			\begin{pmatrix}
				Q \\ P
			\end{pmatrix} =: M \begin{pmatrix}
				Q \\ P
			\end{pmatrix}
		\end{equation}
		and then consider $M$ as the time evolution matrix. The matrix
		\begin{equation}
			L := \begin{pmatrix}
				\frac{P}{2} & \frac{Q}{2} \\ \frac{Q}{2} &-\frac{P}{2}
			\end{pmatrix}
		\end{equation}
		does not commute with $M$, but their commutator is in fact the time evolution of $L$:
		\begin{equation}
			\frac{d}{dt}L = \begin{pmatrix} \dot{P} & \dot{Q} \\ \dot{Q} & - \dot{P} \end{pmatrix} = \begin{pmatrix}
				- \frac{k}{m} Q & \frac{k}{m} P \\
				\frac{k}{m} P & \frac{k}{m} Q
			\end{pmatrix}
			= L M - M L.
		\end{equation}
		Computing the eigenvalues of $L$ via $0 = \det(L- \lambda \mathbb{I}_2) =-(P- \lambda)(P+\lambda) -Q^2$ yields
		\begin{equation}
			 \lambda = \pm \frac{1}{2} \sqrt{Q^2 + P^2} = \pm \sqrt{\frac{k}{4 m} q^2 + \frac{p^2}{4 m^2}} = \pm \sqrt{\frac{1}{2 m}} \sqrt{H},
		\end{equation}
		which means that the eigenvalues are constants of motion, since the Hamiltonian is a constant of motion. Similarly we can recover the Hamiltonian, i.e., the conserved quantity, by computing the trace of the square of the matrix $L$.
	\end{example}
\end{framed}


\clearpage

\section{Lax pairs of the KdV equation}

\noindent
In this section, we will find two distinct Lax pairs for the KdV equation and study the eigenvalue problem of one of the involved operators.
To this aim, let $u: \R \to \R$ be a function, possibly depending on parameters, and consider
$$
L:= \del^2_x + u
$$
seen as operator between suitable subspaces of $\mcF_{par}(\R, \R)$ where $L$ operates on functions $v$ via
$$
Lv = (\del^2_x + u)(v) = v_{xx} + uv,
$$
i.e., $\del_x^2$ is the usual second order differential operator and the zero order term $u$ acts by multiplication. $L$ is thus the 1-dimensional Laplace operator plus a zero order perturbation. Since the Lax formalism varies solutions of an eigenvalue problem we first should determine if the eigenvalue problem  $Lv = \lam v$ can be solved explicitly. There are several ways to approach this problem, e.g.\ by doing an ansatz with formal power series or an ansatz with Fourier series. We follow \cite[Section 1.3]{miwaJimboDate} by using formal power series.

\begin{framed}
\begin{proposition}
 Let $L:= \del^2_x +u$. Then the coefficients in the expansion as formal power series $v=  e^{\sqrt{\lam}  x}\sum_{k\geq 0} \frac{v_k}{ \left(\sqrt{\lam }\ \right)^{k}}$ of the eigenfunctions $v: \R \to \R$ of the eigenvalue problem $Lv = \lam v$ can be determined recursively by solving
 $$
 \del_x v_0=0 \qquad \mbox{and} \qquad \del_x v_k = -\frac{1}{2} \left( \del^2_x v_{k-1} + u v_{k-1}\right) \quad\ \forall\ k \geq 1.
 $$
 Thus $v_0$ is constant and $v_k (x)= -\frac{1}{2} \int^x  \del^2_y v_{k-1}(y) + u(y) v_{k-1}(y) \ dy \ $ $\forall\ k \geq 1$. 
\end{proposition}
\end{framed}

\begin{proof}
If $u: \R \to \R$ vanishes then $Lv=\lam v$ reduces to $v_{xx}=\lam v$ of which the solutions are of the form $v(x)=ce^{\sqrt{\lam}  x}$ with $c \in \R$. If $u \neq 0$ then we make the ansatz
$$v(x) = e^{\sqrt{\lam}  x}\sum_{k\geq 0} \frac{v_k(x)}{ \left(\sqrt{\lam }\ \right)^{k}}.$$
This leads to
\begin{align*}
 0& = v_{xx}+u v-\lam v \\
 & = \left( \sqrt{\lam} e^{\sqrt{\lam}  x}\sum_{k\geq 0} \frac{v_k}{ \left(\sqrt{\lam }\ \right)^{k}} + e^{\sqrt{\lam}  x}\sum_{k\geq 0} \frac{(v_k)_x }{ \left(\sqrt{\lam }\ \right)^{k}} \right)_x
  + u e^{\sqrt{\lam}  x}\sum_{k\geq 0} \frac{v_k}{ \left(\sqrt{\lam }\ \right)^{k}} - \lam e^{\sqrt{\lam}  x}\sum_{k\geq 0} \frac{v_k}{ \left(\sqrt{\lam }\ \right)^{k}} \\
 & = \lam e^{\sqrt{\lam}  x}\sum_{k\geq 0} \frac{v_k}{ \left(\sqrt{\lam }\ \right)^{k}} + 2 \sqrt{\lam} e^{\sqrt{\lam}  x}\sum_{k\geq 0} \frac{(v_k)_x }{ \left(\sqrt{\lam }\ \right)^{k}} + e^{\sqrt{\lam}  x}\sum_{k\geq 0} \frac{(v_k)_{xx} }{ \left(\sqrt{\lam }\ \right)^{k}} \\
 & \quad + u e^{\sqrt{\lam}  x}\sum_{k\geq 0} \frac{v_k}{ \left(\sqrt{\lam }\ \right)^{k}}- \lam e^{\sqrt{\lam}  x}\sum_{k\geq 0} \frac{v_k}{ \left(\sqrt{\lam }\ \right)^{k}} \\
 & = 2 \sqrt{\lam} e^{\sqrt{\lam}  x}\sum_{k\geq 0} \frac{(v_k)_x }{ \left(\sqrt{\lam }\ \right)^{k}} + e^{\sqrt{\lam}  x}\sum_{k\geq 0} \frac{u v_k + (v_k)_{xx} }{ \left(\sqrt{\lam }\ \right)^{k}} 
\end{align*}
After dividing by $e^{\sqrt{\lam}  x} $, we get, for $k \geq 0$, the following equations for the coefficients $v_k$:
 \begin{itemize}[leftmargin=10mm]
  \item 
  $0 = 2 (v_0)_x $ implies $x \mapsto v_0(x)\equiv v_0$ must be constant.
  \item
  $0 = 2 (v_1)_x + u v_0 + (v_0)_{xx} = 2 (v_1)_x + u v_0$ implies $v_1(x) = - \frac{v_0}{2}\int^x u(y) \ dy$.
  \item
  $0 = 2 (v_2)_x + u v_1 + (v_1)_{xx} $ implies $v_2(x) = - \frac{1}{2} \int^x u (y) v_1(y) + (v_1)_{yy} (y) \ dy$.
 \end{itemize}
In general, we get $0 = 2 (v_k)_x + u v_{k-1} + (v_{k-1})_{xx} $ for $k \geq 1$ which leads to 
$$
v_k(x) = - \frac{1}{2} \int^x u (y)v_{k-1}(y) + (v_{k-1})_{yy}(y) \ dy \qquad \forall\ k \geq 1
$$
which proves the claim. 
\end{proof}

Given the operator $L$, we will see that there are (at least) two totally different choices for the operator $M$ to form a Lax pair for the KdV equation. The first one is a differential operator of first order:

\begin{framed}
 \begin{example}
 \label{laxKdvFirst}
Let $u $ be a function and $\ga \in \R$ a constant and 
\begin{equation*}
 \left\{ 
 \begin{aligned}
  L & =\del^2_x + u,  \\
  M &= ( \ga + u_x) - (4 \lam + 2 u) \del_x
 \end{aligned}
 \right. 
\end{equation*}
where $\lam= \lam(t)$ is an eigenvalue of $L$. Abbreviate $L_t:= \del_t L = u_t$. Then $L$ and $M$ form a Lax pair, i.e., we have $L_t + [L,M]=0$, if and only if $u $ solves the KdV equation $u_t+6 uu_x +u_{xxx}=0$. 
 \end{example}
\end{framed}

\begin{proof}
Let $v: \R \times \R \to \R$ be a function in the variables $(x,t) \in \R^2$. We calculate
\begin{align*}
&  (L_t + [L,M])(v) = L_tv +LMv -MLv \\
& \quad = u_t v + (\del_x^2 + u) (( \ga + u_x) - (4 \lam + 2 u) \del_x)(v)
 - (( \ga + u_x) - (4 \lam + 2 u) \del_x)(\del_x^2 + u)(v) \\
& \quad = u_t v + (\del_x^2 + u)( \ga v + u_x v - 4 \lam v_x + 2uv_x)
- (( \ga + u_x) - (4 \lam + 2 u) \del_x) (v_{xx} + uv)  \\
& \quad = u_t v + \ga v_{xx} + u_{xxx} v + 2 u_{xx}v_x + u_x v_{xx} -4\lam v_{xxx} -2 u_{xx} v_x -4 u_x v_{xx} - 2u v_{xxx} \\
&   \qquad\ \ + \ga u v + u u_x v - 4 \lam u v_x - 2 u^2 v_x - \bigl( \ga v_{xx} + u_x v_{xx} -4 \lam v_{xxx} - 2 u v_{xxx} \\
&   \qquad\ \ + \ga u v + u u_x v - 4 \lam u_x v - 4 \lam u v_x - 2 u u_x v - 2 u^2 v_x \bigr) \\
& \quad = u_t v +  u_{xxx} v - 4u_x v_{xx} + 4 \lam u_x v + 2 uu_x v
\end{align*}
and, making use of $\lam$ being an eigenvalue, i.e., $\lam v = Lv = v_{xx} +uv $, and we obtain furthermore
\begin{align*}
& = u_t v +  u_{xxx} v - 4u_x v_{xx} + 4 u_x (v_{xx} +uv) + 2 uu_x v \\
& = u_t v +  u_{xxx} v - 4 u_x v_{xx} + 4 u_x v_{xx} + 4 u u_x v + 2uu_x v \\
& = (u_t + 6 uu_x + u_{xxx}) v.
\end{align*}
Thus $L_t + [L,M]=0$ holds true if and only if $u_t + 6 uu_x + u_{xxx}=0$, i.e., $u$ is a solution of the KdV equation. 
\end{proof}

Now let us consider another possible Lax pair related to the KdV equation where $M$ is an operator of third order.

\begin{framed}
 \begin{example}
 \label{laxKdvThird}
Let $u, c_1, c_2, c_3: \R^2 \to \R $ be functions with variables $(x,t) \in \R^2$. Then the ansatz
\begin{equation*}
 \left\{ 
 \begin{aligned}
  L & =\del^2_x + u,\\
  M &= c_1 \del_x^3 + c_2 \del_x + c_3
 \end{aligned}
 \right. 
\end{equation*}
turns $L$ and $M$ into a Lax pair if and only if
\begin{itemize}
 \item 
 \ $(c_1, c_2, c_3) = (-4, -6u, -3u_x)$,
 \item
 \ $u$ solves the KdV equation $u_t + 6uu_x + u_{xxx} =0$.
\end{itemize}
In other words, $L$ and $M$ form a Lax pair if and only if
 \begin{equation*}
 \left\{ 
 \begin{aligned}
  L & =\del^2_x + u, \\
  M &= -4 \del_x^3 -6u \del_x -3u_x
 \end{aligned}
 \right. 
\end{equation*}
and $u$ satisfies $u_t +6 uu_x + u_{xxx} = 0$.
\end{example}
\end{framed}

\begin{proof}
 Let $u$, $v$, $c_1$, $c_2$, $c_3: \R^2 \to \R$ be functions with variables $(x,t)\in \R^2$ and set $L : =\del^2_x + u$ and $M:= c_1 \del_x^3 + c_2 \del_x + c_3$ and $L_t:= \del_t L =u_t$. We calculate
 \begin{align*}
  (L_t + [L,M])(v)
  &  = u_t v + (\del^2_x + u)(c_1 \del_x^3 + c_2 \del_x + c_3)(v) - (c_1 \del_x^3 + c_2 \del_x + c_3)(\del^2_x + u)(v) \\
  &  = u_t v + (\del^2_x + u) ( c_1 v_{xxx} + c_2 v_x + c_3 v) - (c_1 \del_x^3 + c_2 \del_x + c_3)(v_{xx} + uv) \\
  &  = u_t v + (c_1)_{xx} v_{xxx} + 2 (c_1)_x v_{xxxx} + c_1 v_{xxxxx} + (c_2)_{xx} v_x + 2 (c_2)_x v_{xx} + c_2 v_{xxx}  \\
  & \quad \ \ + (c_3)_{xx} v + 2 (a_3)_x v_x + c_3 v_{xx} + c_1 u v_{xxx} + c_2 u v_x + a_3 u v \\
  & \quad  \ \ - c_1 v_{xxxxx} - c_2 v_{xxx} -c_3 v_{xx} -c_1 u_{xxx}v - 3 c_1 u_{xx} v_x - 3 c_1 u_x v_{xx} \\
  & \quad \ \  - c_1 u v_{xxx} - c_2 u_x v - c_2 u v_x - c_3 u v \\
   &  = u_t v + 2 (c_1)_x v_{xxxx} + (c_1)_{xx} v_{xxx} + 2 (c_2)_x v_{xx} - 3 c_1 u_x v_{xx} + (c_2)_{xx} v_x \\
   & \quad  \ \  - 3c_1 u_{xx} v_x + 2 (c_3)_x v_x + (c_3)_{xx} v - c_1 u_{xxx} v - c_2 u_x v \\
   &  =  2 (c_1)_x v_{xxxx} + (c_1)_{xx} v_{xxx} + \bigl(2 (c_2)_x - 3 c_1 u_x \bigr) v_{xx} \\
   & \quad + \bigl( (c_2)_{xx} - 3c_1 u_{xx}  + 2 (c_3)_x  \bigr) v_x + \bigl( u_t + (c_3)_{xx}  - c_1 u_{xxx}  - c_2 u_x \bigr) v.
 \end{align*}
$0= L_t + [L,M]$ holds true if and only if the coefficients of $v_{xxxx}$, $v_{xxx}$,  $v_{xx}$,  $v_x$, and $v$ vanish. This means
\begin{itemize}[leftmargin=10mm]
 \item 
 $0=2 (c_1)_x$ gives $(c_1)_x =0$ which implies that $c_1$ must be constant. Denote this constant for the moment by $c$. Its precise value will be determined below. 
 \item
 $c_1 \equiv c$ constant satisfies automatically also $(c_1)_{xx}=0$.
 \item
 Moreover, $0= 2 (c_2)_x - 3 c_1 u_x= 2 (c_2)_x - 3 c u_x$ implies $(c_2)_x= \frac{3}{2} c u_x$ so that $c_2(x) = \frac{3}{2} c \int^x u_y(y) \ dy = \frac{3}{2} c u(x)$.
 \item
 Furthermore, $0= (c_2)_{xx} - 3c_1 u_{xx}  + 2 (c_3)_x = \left(\frac{3}{2} c u\right)_{xx} - 3 c u_{xx} + 2 (c_3)_x$ which is equivalent to $(c_3)_x = \frac{3}{4} c u_{xx} $ which implies in turn $c_3= \frac{3}{4} c u_x$.
 \item
 Eventually, we find
 \begin{align*}
   0 & =  u_t + (c_3)_{xx}  - c_1 u_{xxx}  - c_2 u_x = u_t + \frac{3}{4} c u_{xxx} -cu_{xxx} -\frac{3}{2} c u u_x u_t - \frac{1}{4} c u_{xxx}-\frac{3}{2} c u u_x
 \end{align*}
which agrees with the KdV equation $0= u_t + 6 uu_x + u_{xxx}$ if and only if $-4 = c \equiv c_1$.
\end{itemize}
Altogether, we find $c_1 \equiv -4$ and $c_2 = -6u$ and $c_3 = -3 u_x$ which leads to $M= - 4 \del^3_x -6u \del_x -3u_x$. Moreover, we conclude that $L=\del_x^2 +u$ and $M= - 4 \del^3_x -6u \del_x -3u_x$ form a Lax pair if and only if $u$ is a solution of the KdV equation $u_t + 6 uu_x + u_{xxx}=0$.
\end{proof}

The resulting equation and Lax pair (if any) depend quite a bit on the chosen ansatz:

\begin{framed}
 \begin{example}
 \label{laxKdvAlmost}
Let $u, c_1, c_2: \R^2 \to \R $ be functions with variables $(x,t) \in \R^2$. Then the ansatz 
\begin{equation*}
 \left\{ 
 \begin{aligned}
  L & =\del^2_x + u, \\
  M &= \del_x^3 + c_1 \del_x + c_2
 \end{aligned}
 \right. 
\end{equation*}
yields a Lax pair if and only if 
\begin{itemize}
 \item 
 $c_1 = \frac{3}{2}u$ and $c_2 = \frac{3}{4}u_x$,
 \item
 $u$ solves the KdV type equation $u_t -\frac{3}{2} u u_x - \frac{1}{4} u_{xxx}=0$.
\end{itemize}
In other words, $L$ and $M$ are a Lax pair if and only if
\begin{equation*}
 \left\{ 
 \begin{aligned}
  L & =\del^2_x + u, \\
  M &=  \del_x^3 + \frac{3}{2}u \del_x + \frac{3}{4} u_x
 \end{aligned}
 \right. 
\end{equation*}
and $u$ satisfies $u_t -\frac{3}{2} u u_x - \frac{1}{4} u_{xxx}=0$.
\end{example}
\end{framed}

\begin{proof}
Let $u, c_1, c_2: \R^2 \to \R $ be functions with variables $(x,t) \in \R^2$ and $L =\del^2_x + u$ and $M = \del_x^3 + c_1 \del_x + c_2$ and $L_t:= \del_t L = u_t$. We compute
\begin{align*}
 & (L_t + [L,M])(v) \\
 & \quad = u_t v + (\del^2_x + u)(\del_x^3 + c_1 \del_x + c_2)(v) - (\del_x^3 + c_1 \del_x + c_2) (\del^2_x + u)(v) \\
 & \quad =  u_t v + (\del^2_x + u)(v_{xxx} + c_1 v_x + c_2v) - (\del_x^3 + c_1 \del_x + c_2)(v_{xx} + uv) \\
 & \quad =  u_t v + v_{xxxxx} + (c_1 v_x)_{xx} + (c_2 v)_{xx} + u v_{xxx} + c_1 u v_x + c_2 uv \\
 & \qquad \quad\ \ - v_{xxxxx} -(uv)_{xxx} -c_1 v_{xxx} -c_1 (uv)_x - c_2 v_{xx} - c_2 uv \\
 & \quad =  u_t v + (c_1)_{xx} v_x + 2 (c_1)_x v_{xx} + c_1 v_{xxx} + (c_2)_{xx} v + 2 (c_2)_x v_x + c_2 v_{xx} + u v_{xxx}  \\ 
 & \qquad + c_1 u v_x  - u_{xxx}v - 3u_{xx} v_x - 3 u_x v_{xx} - u v_{xxx} - c_1 v_{xxx} - c_1 u_x v - c_1 u v_x - c_2 v_{xx} \\
 & \quad = \bigl( 2 (c_1)_x - 3 u_x \bigr) v_{xx} + \bigl( (c_1)_{xx} + 2 (c_2)_x - 3u_{xx} \bigr) v_x 
 + \bigl( u_t + (c_2)_{xx} - u_{xxx} - c_1 u_x \bigr) v
\end{align*}
$0= L_t + [L,M]$ holds true if and only if the coefficients of $v_{xx}$,  $v_x$, and $v$ vanish. This leads to
\begin{itemize}[leftmargin=10mm]
 \item 
 $ 0= 2 (c_1)_x - 3 u_x$ is equivalent to $(c_1)_x = \frac{3}{2} u_x$ and thus $c_1 =  \frac{3}{2} u$.
 \item
 $ 0 = (c_1)_{xx} + 2 (c_2)_x - 3u_{xx} =  \frac{3}{2} u_{xx} + 2 (c_2)_x - 3u_{xx} =  2 (c_2)_x  -\frac{3}{2} u_{xx}$ is equivalent to $ (c_2)_x = \frac{3}{4} u_{xx}$ which implies $c_2 = \frac{3}{4} u_{x}$.
 \item
 $ 0 = u_t + (c_2)_{xx} - u_{xxx} - c_1 u_x = u_t + \frac{3}{4} u_{xxx} - u_{xxx}- \frac{3}{2} uu_x=u_t- \frac{3}{2} uu_x - \frac{1}{4} u_{xxx} $ is a KdV type equation.
\end{itemize} 
Thus the coefficients are uniquely determined by $c_1 = \frac{3}{2}u$ and $c_2 = \frac{3}{4}u_x$ and the related PDE is of KdV type given by $u_t -\frac{3}{2} u u_x - \frac{1}{4} u_{xxx}=0$. We conclude that $L$ and $M$ form a Lax pair if and only if $c_1 = \frac{3}{2}u$ and $c_2 = \frac{3}{4}u_x$ and $u$ is a solution of $u_t -\frac{3}{2} u u_x - \frac{1}{4} u_{xxx}=0$. 
\end{proof}


\clearpage

\section{Limitations of the Lax formalism}

\noindent
A look at the proofs of \reflaxKdvFirst, \reflaxKdvThird, and \reflaxKdvAlmost\ shows that an ansatz usually leads to several constraints. If the underlying field is commutative one may hope for a reasonable number of terms arising from the commutator in the Lax equation to cancel each other. But as soon as we are not working over the real or complex numbers but the quaternions then much less terms will be able to cancel each other.

\begin{framed}
\begin{definition}
 A linear map $\mcI: \R^{2n}\simeq C^n \to \R^{2n}\simeq C^n$ satisfying $\mcI \circ \mcI = -\Id$ is said to be a {\bf complex structure\index{Complex structure}}. Three complex structures $\mcI_1$, $\mcI_2$, $\mcI_3: \R^{4n} \to \R^{4n}$ form a {\bf quaternionic structure\index{Quaternionic structure}} if $\mcI_1 \circ \mcI_2=-\mcI_2 \circ \mcI_1 = \mcI_3$.
\end{definition}
 \end{framed}

Let us have a look at some examples.

\begin{framed}
 \begin{example}
 Let $i\in \C$ denote the imaginary unit. Then 
 $$\mcI: \R^{2n}\simeq C^n \to \R^{2n}\simeq C^n, \quad z=(z_1, \dots, z_n) \mapsto iz:=(iz_1, \dots, iz_n)$$
 is a complex structure.
 \end{example}
\end{framed}

Furthermore

\begin{framed}
 \begin{definition}
Write elements $x=(x_0, x_1, x_2, x_3) \in \R^4$ as $x=x_0 + ix_1 + j x_2 + k x_3$ where $i, j, k$ satisfy $i^2=j^2=k^2=ijk=-1$ and endow it with the standard addition and the (noncommutative!) multiplication 
\begin{align*}
 x y & = (x_0 + ix_1 + j x_2 + k x_3)(y_0 + iy_1 + j y_2 + k y_3) \\
 & = x_0 y_0 + ix_0y_1 + jx_0 y_2 + k x_0 y_3 +   ix_1y_0 + i^2 x_1 y_1 + ijx_1 y_2 + ik x_1 y_3 \\
 &  \quad + jx_2 y_0 + jix_2 y_1 + j^2 x_2 y_2 + jk x_2 y_3 +  kx_3 y_0 + kix_3 y_1 + kj x_3 y_2 + k^2 x_3 y_3 \\
 & = x_0 y_0 - x_1 y_1 - x_2 y_2 - x_3 y_3 + i(x_0y_1 + x_1 y_0 + x_2 y_3 - x_3 y_2) \\
 & \quad + j( x_0 y_2 + x_2 y_0 -x_1 y_3 + x_3 y_1) + k(x_0 y_3 + x_3 y_0 + x_1 y_2 - x_2 y_2).
\end{align*}
$(\R^4, i, j, k)$ with the above multiplication is usually denoted by $\H$ and its elements are called {\bf quaternions\index{Quaternions}}. $i, j, k \in \H$ are often called {\bf quaternionic units\index{Quaternionic unit}}. {\bf Quaternionic conjugation\index{Quaternionic conjugation}} is defined as $\overline{x}:=x_0  ix_1 - j x_2 - k x_3$.
\end{definition}
\end{framed}

Multiplication in $\H$ is indeed {\em not} commutative as already the very special case $ij=-ji$ shows. In fact, only the subspace of real numbers of $ \H$ commutes with all $z \in \H$.

\begin{framed}
 \begin{example}
  Let $i, j, k \in \H$ be the quaternionic units. Then
  \begin{align*}
  &\mcI : \H^n \to \H^n , & & \mcI (z_1, \dots, z_n)  := (iz_1, \dots, iz_n), \\
  & \mcJ : \H^n \to \H^n  ,  && \mcJ (z_1, \dots, z_n)  := (jz_1, \dots, jz_n), \\ 
  &\mcK: \H^n \to \H^n  ,  & & \mcK (z_1, \dots, z_n)  := (kz_1, \dots, k z_n)
  \end{align*}
are three complex structures that form in fact a quaternionic structure.
\end{example}
\end{framed}

We are now interested in the following PDE.

\begin{framed}
 \begin{definition}
Let $\R^{4n} \simeq \H^n$ be equipped with a quaternionic structure $\mcI_1, \mcI_2, \mcI_3$. Consider $\T^3: = (\R \slash \Z)^3$ with the standard basis and denote the associated partial derivatives by $\del_1$, $\del_2$, $\del_3$. A map $v: \T^3 \to \H^n$ satisfies the {\bf Cauchy-Riemann-Fueter equation\index{Cauchy-Riemann-Fueter equation}} if 
$$
\mcI_1 \del_1 v + \mcI_2 \del_2 v + \mcI_3 \del_3 v =0.
$$
The {\bf Cauchy-Riemann-Fueter operator\index{Cauchy-Riemann-Fueter operator}} is given by
$$
\dirac := \mcI_1 \del_1  + \mcI_2 \del_2  + \mcI_3 \del_3 
$$
and its {\bf quaternionic conjugate\index{Quaternionic conjugate}} by 
$$
\overline{\dirac} := -\mcI_1 \del_1  - \mcI_2 \del_2  - \mcI_3 \del_3 .
$$
The {\bf negative evolution equation\index{Negative evolution equation}} associated with the Cauchy-Riemann-Fueter operator is given by
$$
\del_tu = - \ \dirac (u)
$$
where $u: \T^3 \times \R \to \H $ has coordinates $(x,t) \in \T^3 \times \R $.
 \end{definition}
\end{framed}

Note the following relation between the Cauchy-Riemann-Fueter operator and the Laplace operator:

\begin{framed}
 \begin{remark}
$\dirac$ is a Dirac type operator in the sense that $\dirac \circ \overline{\dirac} = \De$ where $\De $ is the 3-dimensional Laplace operator. 
 \end{remark}
\end{framed}

The Cauchy-Riemann-Fueter equation is of importance for instance in hyperk\"ahler Floer theory where it describes the critical points of the hypersymplectic action functional. Its negative evolution equation coincides with the PDE induced by the negative $L^2$-gradient flow between critical points, for details see \cite{HNS} and its physical interpretation in \cite{BrilleslijperFabert}. The Cauchy-Riemann-Fueter equation can be written as Hamiltonian PDE, see \cite{hohloch}. Thus it makes sense to inquire about the existence of Lax pairs. We will see that the noncommutativity of the multiplication in the target space poses here the real obstacle.

\begin{framed}
 \begin{example}
 \label{exampleNoLax}
 Let $\R^{4n} \simeq \H^n$ be equipped with a quaternionic structure $\mcI_1, \mcI_2, \mcI_3$ and let $\del_1, \del_2, \del_3$ be the partial derivatives on $\T^3$ w.r.t.\ the standard basis. Let $u, v: \T^3 \times \R \to \H^n$ be functions with variables $(x,t) \in \T^3 \times \R$ and set
 \begin{equation*}
  \left\{ 
  \begin{aligned}
L & := \dirac + u,  \\
M & := \overline{\dirac} 
\end{aligned}
  \right.
 \end{equation*}
with $L(v):= \dirac(v) + uv$ where the operation $\H^n \times \H^n \to \H^n$, $(X,Y) \mapsto XY$ must be associative and linear in both factors. 
Then $L_t + [L,M]=0$, i.e., $L$ and $M$ form a Lax pair, if and only if 
\begin{itemize}
 \item 
 $u$ solves $u_t + \dirac(u) =0$,
 \item
 $\mcI_m u = u \mcI_m$ \ for all $m=1, 2, 3$.
\end{itemize}
 \end{example}
\end{framed}

\begin{proof}
 Let $u, v: \T^3 \times \R \to \H^n$ be functions with variables $(x,t) \in \T^3 \times \R$ and let $L  := \dirac + u$ and $M := \overline{\dirac} $ and set $L_t:=\del_t L = u_t$. Since complex structures are linear maps, we get for $m \in \{1, 2, 3\}$ 
 \begin{align*}
  \bigl(\del_1(\mcI_{m} v), \del_2 (\mcI_{m} v), \del_3 (\mcI_{m} v) \bigr) & =   D \left(\mcI_{m} v\right)
  = D \mcI_m|_v.Dv = \mcI_m (\del_1 v, \del_2 v, \del_3 v)  \\
  &  = \left( \mcI_m \del_1 v, \mcI_m  \del_2 v,\mcI_m \del_3 v \right).
 \end{align*}
 Moreover, if $v$ is sufficiently regular, then 
 $$\del^2_{m_1 m_2}v:=\del_{m_1}\del_{ m_2} v = \del_{m_2}\del_{ m_1}v=:\del_{m_2 m_1}^2$$ 
 for all $m_1, m_2 \in \{1, 2,3\}$. We compute
 \begin{align*}
  & (L_t + [L, M])(v) \\
  & \quad = u_tv + (\mcI_1 \del_1  + \mcI_2 \del_2  + \mcI_3 \del_3 +u)( -\mcI_1 \del_1  - \mcI_2 \del_2  - \mcI_3 \del_3)(v) \\
  & \qquad \quad \ \ - ( -\mcI_1 \del_1  - \mcI_2 \del_2  - \mcI_3 \del_3)  (\mcI_1 \del_1  + \mcI_2 \del_2  + \mcI_3 \del_3 + u)(v) \\
  & \quad = u_tv + (\mcI_1 \del_1  + \mcI_2 \del_2  + \mcI_3 \del_3 + u)( -\mcI_1 \del_1 v - \mcI_2 \del_2 v - \mcI_3 \del_3 v) \\
  & \qquad \quad \ \ + (\mcI_1 \del_1  + \mcI_2 \del_2  + \mcI_3 \del_3)(\mcI_1 \del_1 v  + \mcI_2 \del_2 v  + \mcI_3 \del_3 v + uv) \\
  & \quad =  u_tv + \del^2_{11} v - \mcI_1 \mcI_2 \del^2_{12} v - \mcI_1 \mcI_3 \del^2_{13} v - \mcI_2 \mcI_1 \del^2_{12} v + \del^2_{22} v - \mcI_2 \mcI_3 \del^2_{23} \\
  & \qquad \quad \ \ - \mcI_3 \mcI_1 \del^2_{13}v - \mcI_3 \mcI_2 \del^2_{23} v + \del^2_{33} v - u \mcI_1 \del_1 v - u \mcI_2 \del_2 v - u \mcI_3 \del_3 v \\
  & \qquad \quad \ \ - \del^2_{11}v + \mcI_1 \mcI_2 \del^2{12}v + \mcI_1 \mcI_3 \del^2_{13} v + \mcI_1 (\del_1 u) v + \mcI_1 u \ \del_1 v \\
  & \qquad \quad \ \ + \mcI_2 \mcI_1 \del^2_{12}v - \del^2_{22} v + \mcI_2 \mcI_3 \del^2_{23} v + \mcI_2 (\del_2 u) v + \mcI_2 u \ \del_2 v \\
   & \qquad \quad \ \ + \mcI_3 \mcI_1 \del^2_{13} v + \mcI_3 \mcI_2 \del^2_{23} v- \del^2_{33}  + \mcI_3 (\del_3 u) v + \mcI_3 u \ \del_3 v \\
   & \quad =  \bigl( \mcI_1 u - u \mcI_1 \bigr) \del_1v + \bigl( \mcI_2 u - u \mcI_2 \bigr) \del_2 v + \bigl( \mcI_3 u - u \mcI_3 \bigr) \del_3 v \\
    & \qquad \quad \ \ + \bigl(u_t + \mcI_1 \del_1 u  + \mcI_2 \del_2 u + \mcI_3 \del_3 u \bigr) v .
 \end{align*}
The identity $L_t + [L, M]=0$ is thus equivalent with
$$
\left\{
\begin{aligned}
 \ & \mcI_1 u = u \mcI_1, \\
 \ &\mcI_2 u  = u \mcI_2, \\
 \ & \mcI_3 u  = u \mcI_3,  \\
 \ &u_t + \mcI_1 \del_1 u  + \mcI_2 \del_2 u + \mcI_3 \del_3 u =0. 
\end{aligned}
\right.
$$
The first three conditions require $u$ to commute with all three complex structures and the last condition is equivalent to the negative evolution equation $\del_t u = -\dirac (u)$.
\end{proof}

Thus solutions $u$ of the evolution equation $\del_t u = -\dirac (u)$ cannot be very rich in structure if the associated operators should form a Lax pair:

\begin{framed}
 \begin{remark}
 \begin{enumerate}[label=\arabic*),leftmargin=5mm, align=left]
  \item 
   Let $\mcI_1=i$, $\mcI_2=j$, $\mcI_3=k$ be the standard quaternionic structure on $\H$. Then the condition $\mcI_m u = u \mcI_m$ for all $m=1, 2, 3$ in \refexampleNoLax\ is satisfied if and only if $u$ is real valued, i.e., of the form $(u, 0, 0,0) \in \R^4 \simeq \H$.
  \item
  \refexampleNoLax\ was found with the ansatz $M= a \del_1 + b \del_2 + c \del_3$ where $a, b, c: \T^3 \times \R \to \H$ are coefficient functions depending on $(x,t) \in \T^3 \times \R$ and $\mcI_1=i$, $\mcI_2=j$, $\mcI_3=k$ is the standard quaternionic structure on $\H$. When studying the conditions for $L_t + [L,M]=0$, then the coefficient of $v$ is given by $u_t -a \del_1 u - b \del_2 u -c\del_3 u$ which yields $u_t + \dirac(u)=0$ if and only if $a=-i$ and $b=-j$ and $c=-k$.
 \end{enumerate}
 \end{remark}
\end{framed}






\chapter{KdV equation: direct and inverse scattering}

\vspace{-10mm}

In this chapter, we will use the Lax formalism and the scattering technique to solve the initial value problem for the Korteweg–de Vries (KdV) equation.

In the previous chapter, in Examples~\ref{laxKdvFirst} and~\ref{laxKdvThird}, we gave the following two Lax pairs for the KdV equation:
\begin{equation}
	\begin{cases}
		L = \partial^2_{x} + u(x,t), \\
		M = \left(\gamma + u_x(x,t)\right) - \left(4 \lambda + 2 u(x,t)\right) \partial_x  
	\end{cases}
\quad \mbox{and} \qquad
	\begin{cases}
		L = \partial^2_{x} + u(x,t), \\
		M = -4 \partial^3_x - 6 u(x,t) \partial_x - 3 u_x(x,t).
	\end{cases}
\end{equation}
Both pairs satisfy the Lax equation $L_t = [L, M]$.
Roughly speaking, the eigenvalue problem associated with the operator $ L $ allows us to deduce the so-called \emph{scattering data} from the initial condition. The time evolution operator $ M $ then governs the time evolution of these scattering data. Finally, the \emph{inverse scattering transform} (by the use of the Gel'fand-Levitan-Marchenko formula) reconstructs the solution. A diagram of this procedure is shown in Figure~\ref{fig:scatteringinv}.

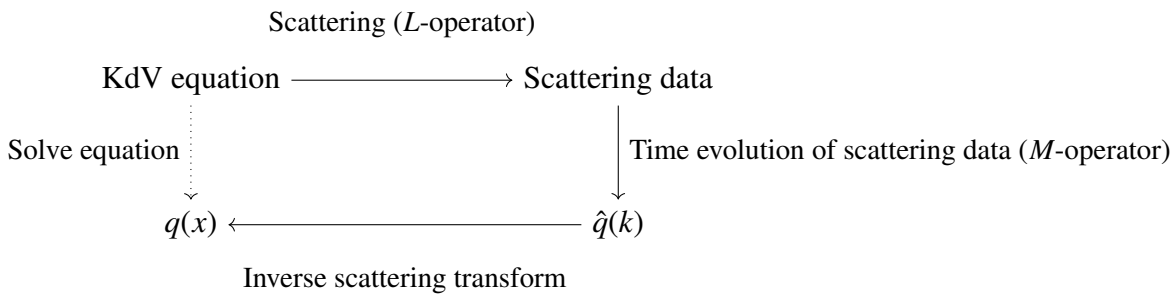
\begin{figure}
	\begin{tikzpicture}
		\matrix (m) [matrix of math nodes, row sep=3em, column sep=7em, nodes={anchor=center}]
		{ 
			\text{KdV equation} & \text{Scattering data} \\
			q(x) & \hat{q}(k) \\
		};
		
		\draw[->] (m-1-1) -- (m-1-2) node[midway, above, font=\small, yshift=1em] {Scattering ($L$-operator)};
		\draw[<-] (m-2-1) -- (m-2-2) node[midway, below, font=\small, yshift=-1em] {Inverse scattering transform};
		
		\draw[->] (m-1-2) -- (m-2-2) node[midway, right, font=\small] {Time evolution of scattering data ($M$-operator)};
		\draw[dotted][->] (m-1-1) -- (m-2-1) node[midway, left, font=\small] {Solve equation};
		
		
	\end{tikzpicture}
	\label{fig:scatteringinv}
	\caption{The procedure around the method of inverse scattering.}
\end{figure}

\vspace{1mm}


\clearpage

\section{Analogy with Fourier transformation}

The abstract procedure of scattering is quite similar to the Fourier transform formalism, though the details differ considerably. The Fourier transform converts a partial differential equation into a new equation for the so-called \emph{Fourier coefficients}, which is hoped to be easier to solve. Then, the inverse Fourier transform allows for reconstructing the solution.

Now recall that, for $f \in L^1(\mathbb{R}^n)$, the Fourier transform is given by
\begin{equation}
	\hat{f}(\mathbf{\xi}) = \frac{1}{\left(2 \pi\right)^{n/2}} \int_{\mathbb{R}^n} f(\textbf{x}) \exp(- i \mathbf{x}\cdot \mathbf{\xi})\, d \mathbf{x}
\end{equation}
and it lies in $C^0_b (\mathbb{R}^n) = \left\{g \in C^0(\mathbb{R}^n) \mid g \text{ bounded}\right\}$. Moreover recall

\begin{framed}
\begin{proposition}
	Let $f\in L^1(\mathbb{R}^n)$ and $\hat{f} \in L^1(\mathbb{R}^n)$. Then
	\begin{equation}
		\tilde{f}(x) = \frac{1}{\sqrt{2 \pi}} \int_{\mathbb{R}^n} \hat{f} e^{i \mathbf{x} \cdot \mathbf{\xi}} d \mathbf{\xi}
	\end{equation}
	coincides with $f$ almost everywhere.
\end{proposition}
\end{framed}

The domain of the Fourier transform can be modified to make it invertible without requiring any specific conditions on the transformed function:

\begin{framed}
\begin{proposition}
	The Fourier transform is a bijection from $S(\mathbb{R}^n)$ (Schwartz space) to itself, where the inverse Fourier transform of a function $\hat{f}(\xi)$ is given by:
	\begin{equation}
		f(x) = \frac{1}{\sqrt{2 \pi}} \int_{\mathbb{R}^n} \hat{f} e^{i \mathbf{x} \cdot \mathbf{\xi}} d \mathbf{\xi}.
	\end{equation}
\end{proposition}
\end{framed}

\begin{proof}
See for instance \cite{reedsimons2} and/or there references therein.
\end{proof}

More properties, technical details, and applications of the Fourier transform can be found in \cite{reedsimons2}.

\begin{framed}
\begin{example}
Let $f, g \in L^1(\mathbb{R}^n, \mathbb{C})$ and denote by $\nabla = \sum_{j=1}^{n} \partial^2_j$ the Laplace operator on $\mathbb{R}^n$ (often called {\em vector Laplacian} in Physics) and consider the following partial differential equation:
	\begin{equation}
		-\nabla f + f = g \label{eqn:pdeexfourier}
	\end{equation}
	Since the Fourier transform turns differentiation into multiplication via $\widehat{\partial_j f} = - i \mathbf{\xi}_j \hat{f}$, applying the Fourier transform to both sides of Equation~\eqref{eqn:pdeexfourier} yields
	\begin{align*}
		\hat{g} & = - \widehat{\sum_{j=1}^{n} \partial_j^2 f } + \hat{f} = -\sum_{j=1}^{n} (i \xi_j )^2 \hat{f} + \hat{f}
		= 	\left(1 + \sum_{j=1}^n \xi_j^2 \right) \hat{f} 
	\end{align*}
	which can be solved for $\hat{f}(\xi)$ via
	\begin{equation}
		\hat{f}(\xi) = \frac{\hat{g}(\xi)}{\left(1 + \sum_{j=1}^n \xi_j^2 \right)}.
	\end{equation}
	Applying the inverse Fourier transform then yields the solution of the original Equation \eqref{eqn:pdeexfourier}:
	\begin{equation}
		f(\mathbf{x}) = \frac{1}{(2 \pi)^{n/2}} \int_{\mathbb{R}^n} \frac{\hat{g}(\mathbf{\xi}) \exp\left(i \mathbf{x} \cdot \mathbf{\xi} \right)}{\left(1 + \sum_{j=1}^n \xi_j^2 \right)} d \mathbf{\xi}.
	\end{equation}
\end{example}
\end{framed}

\begin{framed}
\begin{example}[Heat equation]
	\label{ex:heat}
	Let $g \in L^2(\mathbb{R})$ and consider the heat equation
	\begin{equation}
		\label{eqn:heateqn}
		\begin{cases}
			\partial_t f(t,x) = \partial_x^2 f(t,x), \\
			f(0,x) = g(x).
		\end{cases}
	\end{equation}
	 Note that requiring $g \in L^2(\mathbb{R})$ ensures that the Fourier transform and inverse Fourier transform are both bijective mappings from $L^2(\mathbb{R})$ to itself. To solve the heat equation we will follow the procedure in Figure~\ref{fig:heatfourier}: Using the inverse Fourier transform, we rewrite
	\begin{equation}
	\label{eqn:fourierheat}
		f(t,x) = \frac{1}{\sqrt{2 \pi}} \int_{-\infty}^{+\infty} \hat{f}(t,k) \, e^{i k x} dk
	\quad \mbox{and} \quad
		f(0,x) = g(x) = \frac{1}{\sqrt{2 \pi}} \int_{-\infty}^{+\infty} \hat{g}(k) e^{i k x} dk.
	\end{equation}
	Substituting this into the heat equation~\eqref{eqn:heateqn}, we obtain the new equation
	\begin{equation}
		\begin{cases}
			\partial_t \hat{f}(t,k) = - k^2 \hat{f}(t,k), \\
			\hat{f}(0,k) = \hat{g}(k)
		\end{cases}	
	\end{equation}
	which is an ordinary differential equation whose solution is given by
	\begin{equation}
		\hat{f}(t,k) = \hat{g}(k) e^{- k^2 t}. \label{eqn:heatfouriersolution}
	\end{equation}
	Putting this solution into the left formula in \eqref{eqn:fourierheat} gives
	\begin{equation}
		f(t,x) = \frac{1}{\sqrt{2 \pi}} \int_{-\infty}^{+\infty} \hat{g}(k)\, e^{- k^2 t + i k x} dk
	\end{equation}
	from which we obtain the solution for the original equation:
	\begin{equation}
		f(t,x) = \frac{1}{2 \pi} \int_{-\infty}^{\infty} \int_{-\infty}^{+\infty} g(y) e^{-k^2 t + i k (x-y)} dy dk.
	\end{equation}
\end{example}
\end{framed}

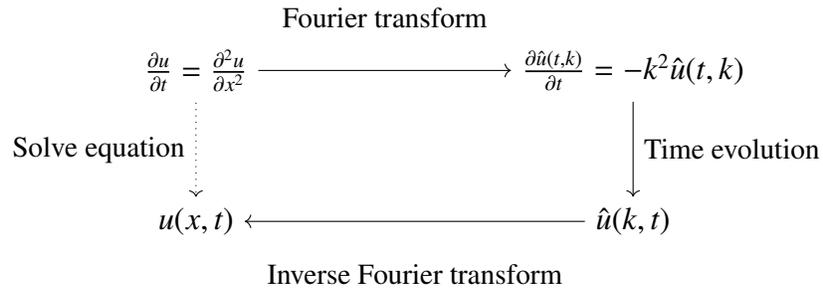
\begin{figure}
	\begin{tikzpicture}
		\matrix (m) [matrix of math nodes, row sep=3em, column sep=8em, nodes={anchor=center}]
		{ 
			\frac{\partial u}{\partial t} = \frac{\partial^2 u}{\partial x^2} & \frac{\partial \hat{u}(t,k)}{\partial t} = - k^2 \hat{u}(t,k) \\
			u(x,t) & \hat{u}(k,t) \\
		};
		
		\draw[<-] (m-2-1) -- (m-2-2) node[midway, below, font=\small,yshift=-1em] {Inverse Fourier transform};
		
		\draw[->] (m-1-2) -- (m-2-2)  node[midway, right, font=\small] {Time evolution};
		\draw[dotted][->] (m-1-1) -- (m-2-1) node[midway, left, font=\small] {Solve equation};
		
		\draw[->] (m-1-1) -- (m-1-2) node[midway, above, font=\small, yshift=1em] {Fourier transform};
		
	\end{tikzpicture}
	\caption{The procedure around the Fourier transform method.}
	\label{fig:heatfourier}
\end{figure}


\clearpage

\section{Motivation from physics}
Scattering is a well-studied framework in physics, as a significant portion of modern particle physics relies on the analogy between particle interactions and scattering phenomena. In this section, we will give an introductory overview of quantum mechanics and scattering theory to clarify some aspects of the physics behind the scattering scheme.

For simplicity, we will restrict our study to a one-dimensional domain $\mathbb{R}$. Note that we will also consider the reduced Planck constant equal to unity, i.e., $\hbar = 1$. For a more comprehensive treatment of quantum theory from a physical perspective, see~\cite{qmsakurai}, and for a mathematical approach, refer to \cite{qmhall}. We now describe the postulates of quantum mechanics in the Schr\"odinger representation.

Typically, the state of a quantum system at a fixed time is identified by a vector in a Hilbert space, which in our case is $L^2(\mathbb{R})$. However, when studying scattering phenomena, the set of possible functions is extended, allowing for non-normalizable functions as well.

An \emph{observable} refers to a measurable quantity of the physical system, represented by a self-adjoint operator acting on the Hilbert space. The possible outcomes of a measurement are the eigenvalues of the corresponding self-adjoint operator. When a measurement is made, the wave function collapses into the eigenfunction associated with the measurement outcome.

The time evolution of a quantum state is governed by the \emph{time-dependent Schr\"odinger equation}
\begin{equation}
	- i\partial_t \psi(t,x) = \hat{H} \psi(t,x)
\end{equation}
where $\hat{H}$ is the \emph{Hamiltonian operator} (representing the energy operator), often briefly called the Hamiltonian of the system, and $\psi(t,x)$ is the function representing the quantum state.

In general, the Hamiltonian operator of a one-dimensional system takes the form
\begin{equation}
	H = - \frac{\partial^2_x}{2 m} + V(x)
\end{equation}
where the first term represents the kinetic energy, i.e., the square of the momentum operator
$$	\hat{p} = -i \partial_x,$$
and the second term represents the potential energy $V(x)$. The Hamiltonian operator acts on functions $\psi$ via $H \psi = -\del^2_x \psi + V(x) \psi$.

The goal from a physical perspective is to study an obstacle, represented mathematically by a potential $V(x)$, through the asymptotic behavior of a test wave (or particle beam). This is done by sending a beam from $-\infty$, which collides with the potential creating a reflected beam, and a transmitted beam which is the part of the beam that overcomes the obstacle.

The setup thus results in three beams:
\begin{itemize}
	\item To the left of the obstacle, as $x \to -\infty$, we have the incident beam, represented by the wave function $\psi_{inc} \sim \exp \{ + i k x\}$, and a reflected beam, with the wave function $\psi_{ref} \sim \exp\{-i k x\}$.
	\item To the right of the obstacle, as $x \to +\infty$, we have the transmitted beam $\psi_{tra} \sim \exp\{- i k x\}$.
\end{itemize}

\begin{framed}
\begin{remark}
	The beams are mathematically described by $\psi_{inc} \sim \exp\{+ i k x\}$, moving from left to right. In the time-dependent version
	\begin{equation}
		\tilde{\psi}_{inc} \sim \exp\{- \omega t + i k x\},
	\end{equation}
	the wave moves from left to right. At any fixed time, these wave functions are not normalizable, i.e., they do not belong to $L^2(\mathbb{R})$. For this reason, we call them {\em free states}, and they always have a positive energy value.
	In contrast, states in $L^2(\mathbb{R})$ are called {\em bound (or localized) states}, as any $L^2(\mathbb{R})$ function can be approximated by a compactly supported function. These functions are often associated with a negative energy value.
\end{remark}
\end{framed}

The potentials studied are of two types:
\begin{itemize}
	\item Rapidly decaying potential:
	\begin{equation}
		V(x) \underset{\abs{x}\to\infty}{\longrightarrow} 0
	\end{equation}
	\item Finite-range potential, i.e., a potential that is non-zero only within a compact interval:
	\begin{equation}
		V(x) = \begin{cases}
			\tilde{V}(x), \qquad x \in\ ]a,b[, \\
			0, \quad \  \qquad x \notin\ ]a,b[.
		\end{cases}
	\end{equation}
\end{itemize}

The choice of these types of potential is not arbitrary: the first type represents long-range, non-localized interactions such as the electromagnetic potential, while the second type can represent either a particle itself or a finite-range potential, such as that created by the strong force.

\clearpage


\section{Scattering procedure}
In this section, we will derive the set of scattering data by approaching the Korteweg–de Vries (KdV) equation as a scattering problem. We will examine the $L$-operator associated with the KdV equation and show that its asymptotic analysis corresponds to the previously illustrated physical setting.

Let us begin by examining the eigenvalue problem associated with the operator $L$ in a time-independent context, as we aim to construct the scattering data at $t=0$. This eigenvalue problem corresponds to the Schr\"odinger-type eigenvalue equation
\begin{equation}
	L v = \left(\partial^2_x + u(x)\right) v = k^2 v
\end{equation}
where $k^2 \in \mathbb{R}$ is an eigenvalue.
We study this eigenvalue problem w.r.t.\ the two limits $x \to -\infty$ and $x \to +\infty$ and under the Faddeev condition \cite{kdvfaddev} $u \in P_1$ where
\begin{equation*}
	 P_1 : = \left\{ g\in C^2(\mathbb{R})  \ \left| \ \int_{\mathbb{R}} (1+ \lvert x \rvert) \lvert g(x)\rvert dx < +\infty \right. \right\}.
\end{equation*}
This implies $u \to 0$ as $\lvert x \rvert \to +\infty$. More details over this setting can be found for instance in \cite{ablowitzClarkson}.
Under these conditions, the Schr\"odinger problem simplifies to
\begin{equation}
	\partial_x^2 v = k^2 v
\end{equation}
for both limits $x \to -\infty$ and $x \to +\infty$. Thus, we identify two sets of eigenfunctions (indexed by the eigenvalue in question):
\begin{equation}
	\begin{cases}
		\{\phi_k(x) , \overline{\phi}_k(x)\} \qquad \text{for} \quad x \to +\infty, \\
		\{\psi_k(x), \overline{\psi}_k(x)\} \qquad \text{for} \quad x \to -\infty
	\end{cases}
\end{equation}
with the following asymptotic behaviors:
\begin{equation}
\label{eqn:asymptotic}
	\begin{matrix}
		\phi_k(x) \underset{x \to +\infty}{\sim} e^{-i k x}, & && \overline{\phi}_k(x) \underset{x \to +\infty}{\sim} e^{i k x}, \\
		\psi_k(x) \underset{x \to -\infty}{\sim} e^{i k x}, & & & \overline{\psi}_k(x) \underset{x \to -\infty}{\sim} e^{-i k x}.
	\end{matrix}
\end{equation}
This implies in particular
\begin{equation}
\label{eqn:invphipsi}
	\begin{cases}
		\phi_k(x) = \overline{\phi}_{-k}(x), \\
		\psi_k(x) = \overline{\psi}_{-k}(x).
	\end{cases}
\end{equation}
Furthermore, as both sets satisfy the same differential equation, we can express them as a linear combination
\begin{equation}
	\begin{cases}
		\phi_k(x) = A(k) \psi_k(x) + B(k) \overline{\psi}_k(x), \\
		\overline{\phi}_k(x) = C(k) \psi_k(x) + D(k) \overline{\psi}_k(x)
	\end{cases}
\end{equation}
and, by applying Equation~\eqref{eqn:invphipsi}, we obtain
\begin{align}
	\phi_k(x) &= \overline{\phi}_{-k}(x), \\
	A(k) \psi_k(x) + B(k) \overline{\psi}_k(x) &= C(-k) \psi_{-k}(x) + D(-k) \overline{\psi}_{-k}(x), \\
	A(k) \psi_k(x) + B(k) \overline{\psi}_k(x) &= C(-k) \overline{\psi}_{k}(x) + D(-k) \psi_{k}(x).
\end{align}
This leads us to the conditions
\begin{equation}
	\begin{cases}
		B(k) = C(-k), \\
		A(k) = D(-k)
	\end{cases}
\end{equation}
which lead to
\begin{equation}
	\label{eqn:equationsysincref}
	\begin{cases}
		\phi_k(x) = A(k) \psi_k(x) + B(k) \overline{\psi}_k(x), \\
		\overline{\phi}_k(x) = B(-k) \psi_k(x) + A(-k) \overline{\psi}_k(x).
	\end{cases}
\end{equation}
To better characterize the coefficients $A(k)$ and $B(k)$, we introduce

\begin{framed}
\begin{definition}
	Let $f,g \in C^1(\mathbb{R},\mathbb{C})$. Then the {\em Wronskian} is the skew-symmetric, bilinear operator
	\begin{equation}
		W(f,g) := f \partial_x g - g \partial_x f.
	\end{equation}
	More generally, the Wronskian for $n$ functions $f_1, \ldots, f_n \in C^{n-1}(\mathbb{R},\mathbb{C})$ is defined as
	\begin{equation}
		W(f_1, \ldots, f_n) := \det \begin{pmatrix}
			f_1 & f_2 & \cdots & f_n \\
			\partial_x f_1 & \partial_x f_2 &  \cdots & \partial_x f_n \\
			\partial^2_x f_1 & \partial_x^2 f_2 & \cdots & \partial_x^2 f_n\\
			\vdots & \vdots  & \ddots & \vdots \\
			\partial_x^{n-1} f_1 & \partial_x^{n-1} f_2 & \cdots & \partial_x^{n-1} f_n
		\end{pmatrix}.
	\end{equation}
\end{definition}
\end{framed}


\begin{framed}
\begin{lemma}
	\label{lemma:Wronskian}
	Let $f,g \in C^2(\mathbb{R})$ satisfy the same Schrödinger equation. Then
	\begin{equation}
		\partial_x W(f,g) = f \partial_x^2 g - g \partial_x^2 f =0.
	\end{equation}
\end{lemma}
\end{framed}

\begin{proof}
	Let $e\in\mathbb{R}$ and consider a Schrödinger equation with a potential $u(x) \in P_1$,
	\begin{equation}
	\label{eqn:generalSchroedinger}
		-\partial_x^2 f + u(x) f = e f.
	\end{equation}
	Let $f,g$ be two different solutions of \eqref{eqn:generalSchroedinger}. Then computing the Wronskian yields
	\begin{align*}
		\partial_x W(f,g) &= \partial_x \left(f \partial g - g \partial f\right) = \left(f \partial_x^2 g - g \partial_x^2 f\right)= f \left( u(x) -e\right) g - g \left( u(x) -e\right) f = 0.
	\end{align*}
\end{proof}

Since the Wronskian of $\phi_k$ and $\phi_{-k}$ is constant in $x$ by Lemma~\ref{lemma:Wronskian}, we can compute it by considering the asymptotic behavior of $\phi_k$ and $\phi_{-k}$ as $x\to+\infty$ and get
\begin{align}
	W(\phi_k(x),\overline{\phi}_k(x)) = 2 i k.
\end{align}
The same holds for $W(\psi_k(x), \overline{\psi}_k(x))$.
If we compute the Wronskians using relations~\eqref{eqn:equationsysincref} we obtain
\begin{align}
	W(\phi_k, \overline{\phi}_k(x)) &= W(A(k) \psi_k(x) + B(k) \overline{\psi}_k(x), B(-k) \psi_k(x) + A(-k) \overline{\psi}_k(x))\\
	&=A(k) A(-k) W(\psi_{k}(x), \overline{\psi}_k(x)) + B(k) B(-k) W(\overline{\psi}_k(x),{\psi}_k(x)) \\
	&=\left(A(k) A(-k) -B(k) B(-k) \right) W(\psi_{k}(x), \overline{\psi}_k(x))
\end{align}
which leads to
\begin{equation}
 \label{eqn:coefftransreflc}
	A(k) A(-k) - B(k) B(-k) = 1.
\end{equation}
Moreover, one can express $A(k)$ and $B(k)$ by means of the Wronskian since
\begin{align*}
	W(\phi_k(x),\psi_k(x)) &= B(k) W(\overline{\psi}_k(x), \psi_k(x)) = - 2 i k B(k)
\end{align*}
leads to
$$
B(k) =  - \frac{W(\phi_k(x),\psi_k(x))}{2 i k}
$$
and
\begin{align*}
	W(\phi_k(x),\overline{\psi}_k(x)) &= A(k) W(\psi_k(x), \overline{\psi}_k(x)) = 2 i k A(k)
\end{align*}
leads to
$$
A(k) =  \frac{W(\phi_k(x),\overline{\psi}_k(x))}{2 i k}.
$$
To simplify the notation in the following computations we set
\begin{align*}
	M_k(x) := \phi_k(x) \exp(i k x), && N_k(x) := \psi_k(x) \exp(i k x), \\
	\overline{M}_k(x) := \overline{\phi}_k(x) \exp(i k x), && 	\overline{N}_k(x) := \overline{\psi}_k(x) \exp(i k x)
\end{align*}
and obtain the asymptotic expansions
\begin{equation}
	\label{eqn:asymptMN}
	\begin{aligned}
		& M_k(x) \underset{x \to +\infty}{\sim} 1 ,\qquad \quad  && N_k(x) \underset{x \to -\infty}{\sim} 1, \\
		& \overline{M}_k(x) \underset{x \to +\infty}{\sim} e^{2 i k x}, \qquad \quad  && \overline{N}_k(x) \underset{x \to -\infty}{\sim} e^{2 i k x}
	\end{aligned}
\end{equation}
with the relations
\begin{equation}
	\begin{dcases}
		N_k(x) = \overline{N}_{-k}(x) e^{2 i k x}, \label{eqnNforNbar}\\
		M_k(x) = \overline{M}_{-k}(x) e^{ 2 i k x}.
	\end{dcases}
\end{equation}
Important coefficients are the ratios between $B(k)$ and $A(k)$ and their complex conjugates. We define the \emph{transmission coefficient} $\tau$ and the \emph{reflection coefficients} $\rho$ by
\begin{equation}
	\tau(k) := \frac{1}{B(k)} \qquad \mbox{and} \qquad  \rho(k):=\frac{A(k)}{B(k)}. \label{eqn:defrefandtrans}
\end{equation}
Their complex conjugates are given by
\begin{equation}
	\overline{\tau}(k) = \frac{1}{\overline{B}(k)}\qquad \mbox{and} \qquad \overline{\rho}(k) = \frac{A(-k)}{B(-k)}.
\end{equation}
Moreover, using equation~\eqref{eqn:coefftransreflc}, we obtain $ 1 - \lvert\rho(k)\rvert^2 = \lvert \tau(k) \rvert^2 $ and thus
\begin{equation}
	\lvert\rho(k)\rvert^2 + \lvert \tau(k) \rvert^2 = 1.
\end{equation}
Equations~\eqref{eqn:equationsysincref}, the definition of the transmission coefficients, and recalling Equation~\eqref{eqnNforNbar} together yield
\begin{equation}
	\frac{M_k(x)}{B(k)} = \overline{N}_k(x) + \rho(k) e^{2 i k x} \overline{N}_{-k}(x). \label{eqn:moverA}
\end{equation}
The following statement can be found for instance in \cite{ablowitzClarkson}.

\begin{framed}
\begin{proposition}
	\label{lemma:221}
	Assume the above described setting. Then
	\begin{enumerate}
		\item $M_k(x)$ and $B(k)$ can be expanded analytically in the half-plane $\Im(k)>0$, and
		\begin{equation}
			\begin{dcases}
				\lim_{\lvert k\rvert \to\infty} M_k(x) \to 1,
				\lim_{\lvert k\rvert \to\infty} a(x) \to 1.
			\end{dcases}
		\end{equation}
		\item $N_k(x)$ can be extended to the half-plane $\Im(k)<0$, and
		\begin{equation}
			\lim_{\lvert k \rvert \to \infty} N_k(x) \to 1
		\end{equation}
	\end{enumerate}
\end{proposition}
\end{framed}

\begin{proof}
	Consider
	\begin{equation*}
		\partial_x^2 v(x) + u(x) v(x) = k^2 v(x)
	\end{equation*}
	and perform the transformation $v_k(x) = m_k(x) \exp(- i k x)$ and obtain
	\begin{equation}
	\label{eqn:Shro1lemma}
		\partial_x^2 m_k(x) - 2 i k \partial_x m_k(x) = - u(x) m_k(x).
	\end{equation}
	To solve equation~\eqref{eqn:Shro1lemma}, we will study the Green function associated to this differential equation, i.e., the function $G_k(x)$ which solves
	\begin{equation*}
		\partial_x^2 G_k(x) - 2 i k  \partial_x G_k(x) = - \delta(x)
	\end{equation*}
	where $\delta(x)$ is the Dirac delta function.
	This allows to write the solution of Equation~\eqref{eqn:Shro1lemma} as
	\begin{equation}
		m_k(x) = 1 +\int_{-\infty}^{+\infty} G_k(x-\xi) u(\xi) m_k(\xi) d\xi
	\end{equation}
	if for $\vert x \rvert \to +\infty$ the function $m(x) \to 1$.
	By means of the Fourier transform we obtain
	\begin{align*}
		\partial_x^2 \int_{-\infty}^{+\infty} e^{i l x} \hat{G}_k(l) d l - 2 i k \partial_x \int_{-\infty}^{+\infty} e^{i l x} \hat{G}_k(l) d l + \int_{-\infty}^{+\infty} e^{i l x} d l
		\int_{-\infty}^{+\infty} \left((i l)^2 \hat{G}_k(l) - 2 i k (i l) \hat{G}_k(l) + 1\right) e^{i l x} d l = 0
	\end{align*}
	which leads to
	\begin{equation}
		\hat{G}_k(l) = \frac{1}{-2 l k + l^2}.
	\end{equation}
	By performing the inverse Fourier transform we get
	\begin{equation}
		\frac{1}{\sqrt{2 \pi}} \int_\gamma \frac{e^{i l x}}{-2 l k + l^2} dl.
	\end{equation}
	For this inverse Fourier transform we integrate over a closed path $\gamma$, as one in Figure~\ref{fig:polesinte} and closed by a half circle.
	Since the poles of $\hat{G}_k{l}$ are $l=0$ and $l=2k$, and both lie on the real line. The choice of the path gamma is not unique since there are two possible contours $\gamma_+$ and $\gamma_{-}$ as sketched in Figure~\ref{fig:polesinte}.
	
	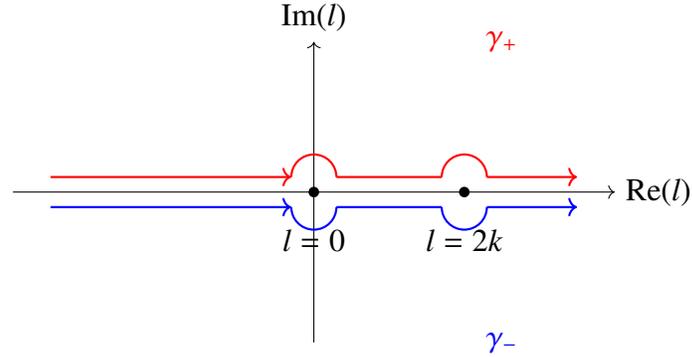
\begin{figure}
		\begin{tikzpicture}
			\draw[->] (-4, 0) -- (4, 0) node[anchor=west] {Re$(l)$};
			
			\draw[->] (0, -2) -- (0, 2) node[anchor=south] {Im$(l)$};
			
			\fill (0, 0) circle (2pt) node[below=10pt] {$l=0$}; 
			\fill (2, 0) circle (2pt) node[below=10pt] {$l=2k$}; 
			
			\draw[thick, red, ->] (-3.5, 0.2) -- (-0.3, 0.2);
			\draw[thick, red] (0.3, 0.2) -- (1.7, 0.2);
			\draw[thick, red, ->] (2.3, 0.2) -- (3.5, 0.2);
			
			\draw[thick, red] (0.3, 0.2) arc[start angle=0, end angle=180, radius=0.3];
			\draw[thick, red] (2.3, 0.2) arc[start angle=0, end angle=180, radius=0.3];
			
			\draw[thick, blue, ->] (-3.5, -0.2) -- (-0.3, -0.2);
			\draw[thick, blue] (0.3, -0.2) -- (1.7, -0.2);
			\draw[thick, blue, ->] (2.3, -0.2) -- (3.5, -0.2);
			
			\draw[thick, blue] (0.3, -0.2) arc[start angle=0, end angle=-180, radius=0.3];
			\draw[thick, blue] (2.3, -0.2) arc[start angle=0, end angle=-180, radius=0.3];
			
			\node at (2.5, 2) {\textcolor{red}{$\gamma_{+}$}};
			\node at (2.5, -2) {\textcolor{blue}{$\gamma_{-}$}};
		\end{tikzpicture}
		\label{fig:polesinte}
		\caption{Integration paths: $\gamma_{+}$ and $\gamma_{-}$}
	\end{figure}
	The two paths leads to two different Green functions since the paths of integration in order to obtain the Green functions depends on the situation and there is some freedom in the choice (in the literature, this is often referred to as \emph{prescriptions}; in physics literature the Green functions are also called {\em propagators}). Let $\Theta(x)$ denote the Heaviside function. Then
	\begin{align}
		\gamma_+ & \qquad \mbox{gives rise to} \qquad  G^{+}_k(x) = \frac{1}{2 i k}  \left( 1 - e^{2 i k x}\right) \Theta(x), \\
		\gamma_- &\qquad \mbox{gives rise to} \qquad  G^{-}_k(x) = \frac{-1}{2 i k}  \left( 1 - e^{2 i k x}\right) \Theta(-x).
	\end{align}
	Note that both Green functions vanishes for $\lvert k \rvert \to +\infty$.
	
	Now we use these two Green functions to obtain two integral expressions (Volterra integral equations), one for $M_k(x)$ and one for $N_k(x)$:
	\begin{equation}
	\label{eqn:stepinvscat1}
	\begin{aligned}
		M_k(x) &= 1 + \frac{1}{2 i k} \int_{-\infty}^x \left(1 - e^{2 i k (x-y)}\right) u(y) M_k(y) dy, \\
		\overline{N}_k(x) &= 1 - \frac{1}{2 i k} \int_{x}^{+\infty} \left(1 - e^{2 i k (x-y)}\right) u(y) \overline{N}_k(y) dy.
	\end{aligned}
	\end{equation}
	Since the solution of these integral equations is unique and both the Green functions are analytic, we can conclude that $M_k(x)$ is analytic for $
	\Im(k)>0$ and $\overline{N}_k(x)$ is analytic for $\Im(k)<0$.
	
	To study the analytic property of $B(k)$ we recall that
	\begin{equation}
		\frac{M_k(x)}{B(k)} = \overline{N}_k(x) + \rho(k) e^{2 i k x} \overline{N}_{-k}(x)
	\end{equation}
	that can be transformed into
	\begin{equation}
	\label{eqn:transformationmkx}
		M_k(x) - B(k) \overline{N}_k(x) = A(k) e^{2 i k x} \overline{N}_{-k}(x) .
	\end{equation}
	By using the integral equations~\eqref{eqn:stepinvscat1}, we obtain for the left side of Equation \eqref{eqn:transformationmkx}:
	\begin{align*}
		& M_k(x) - B(k) \overline{N}_k(x) \\
		& =   \left(1 + \frac{1}{2 i k} \int_{-\infty}^x \left(1 - e^{2 i k (x-y)}\right) u(y) M_k(y) dy\right) - B(k)\left(1 - \frac{1}{2 i k} \int_{x}^{+\infty} \left(1 - e^{2 i k (x-y)}\right) u(y) \overline{N}_k(y) dy \right) \\
		 & = 1 - B(k) + \frac{1}{2i k} \left(\int_{-\infty}^x \left(1 - e^{2 i k (x-y)}\right) u(y) M_k(y) dy \right)
		 + \frac{B(k)}{2 i k} \left(\int_{x}^{+\infty} \left(1 - e^{2 i k (x-y)}\right) u(y) \overline{N}_k(y) dy \right)  \\
		 & = 1 - B(k) + \frac{1}{2i k} \left(\int_{-\infty}^{+\infty} \left(1 - e^{2 i k (x-y)}\right) u(y) M_k(y) dy \right)  - \frac{1}{2i k} \left(\int_{x}^{+\infty} \left(1 - e^{2 i k (x-y)}\right) u(y) M_k(y) dy \right) \\
		 & \quad  + \frac{B(k)}{2 i k} \left(\int_{x}^{+\infty} \left(1 - e^{2 i k (x-y)}\right) u(y) \overline{N}_k(y) dy \right)  \\
		 & = 1 - B(k) + \frac{1}{2i k} \left(\int_{-\infty}^{+\infty} \left(1 - e^{2 i k (x-y)}\right) u(y) M_k(y) dy \right) \\
		 & \quad + \frac{1}{2 i k} \left(\int_{x}^{+\infty} \left(1 - e^{2 i k (x-y)}\right) u(y) \left( B(k) \overline{N}_k(y) - M_k(y) \right) dy \right)\\
		 & = 1 - B(k) + \frac{1}{2i k} \left(\int_{-\infty}^{+\infty} \left(1 - e^{2 i k (x-y)}\right) u(y) M_k(y) dy \right) \\
		 & \quad  - \frac{1}{2 i k} \left(\int_{x}^{+\infty} \left(1 - e^{2 i k (x-y)}\right) u(y) \left( M_k(y) - B(k) \overline{N}_k(y)  \right) dy \right).
	\end{align*}
	Now consider the right-hand side of Equation~\eqref{eqn:transformationmkx}:
	\begin{align*}
		A(k) e^{2 i k x} \overline{N}_{-k}(x) & =  A(k) e^{2 i k x} \left(1 - \frac{1}{2 i k} \int_{x}^{+\infty} \left(1 - e^{2 i k (x-y)}\right) u(y) \overline{N}_{-k}(y) dy \right) \\
		& =  A(k) e^{2 i k x} - \frac{1}{2 i k} \left(\int_{x}^{+\infty} \left(1 - e^{-2 i k (x-y)}\right) u(y) A(k) e^{2 i k x} \overline{N}_{-k}(y) dy \right).
	\end{align*}
	By putting all together, and using~\eqref{eqn:transformationmkx} again we obtain
	\begin{align}
		1 - B(k) + \frac{1}{2i k} \left(\int_{-\infty}^{+\infty} \left(1 - e^{2 i k (x-y)}\right) u(y) M_k(y) dy \right) = A(k) e^{2 i k x}.
	\end{align}
	By comparing the coefficients we eventually obtain
	\begin{equation}
		\begin{dcases}
			B(k) = 1 + \frac{1}{2 i k} \int_{-\infty}^{\infty} u(y) M_k(y) dy , \\
			A(k) = - \frac{1}{2 i k} \int_{-\infty}^{\infty} u(y) M_k(y) e^{- 2 i k y} dy.
		\end{dcases} \label{eqn:recoverAB}
	\end{equation}
	For the integral representations we get the analytic property of $B(k)$ for $\Im(k) > 0$, the limit for the same conditions:
	\begin{equation}
		\lim_{\lvert k \rvert \to +\infty} B(k) \to 1.
	\end{equation}
	\end{proof}
	
	Now we show

	\begin{framed}
	\begin{proposition}
		\label{prop:secondprop}
		The function $B(k)$ can have only a finite number of simple zeros located on the upper part of the imaginary axis.
	\end{proposition}
	\end{framed}
	
	\begin{proof}
		We will prove the statement in three steps.

		{\bf Step 1:}
		Let $k_0$ be a zero of $B(k)$, with $k_0 = \xi_0 + i \delta_0$ (where $\xi_0, \delta_0 \in \mathbb{R}$). Then by using Equation~\eqref{eqn:equationsysincref} we get
		\begin{equation}
			\phi(x,k_0) = B(k_0) \psi_k(x)
		\end{equation}
		en subsequently
		\begin{equation}
			\label{eqn:asympa(k)zeros}
			\begin{dcases}
				\phi_{k_0}(x) \underset{x \to -\infty}{\sim} B(k_0) e^{- i \xi_0 x} e^{\delta_0 x}, \\
				\phi_{k_0}(x) \underset{x \to \infty}{\sim} B(k_0) e^{+ i \xi_0 x} e^{-\delta_0 x}.
			\end{dcases}
		\end{equation}
		Since $\phi_k(x)$ is square integrable we conclude $\delta_0 > 0$.
		Moreover, since $\phi_k(x)$ satisfies the Schr\"odinger equation
		\begin{equation}
		\label{eqn:Schorosimple}
			\partial_x^2 \phi_k(x) + \left(u(x) + k^2 \right) \phi_k(x) = 0
		\end{equation}
		and denoting by $\phi^*_k(x)$ the complex conjugate of $\phi_k(x)$ we compute
		\begin{align*}
			\partial_x \left(W(\phi_k(x), \phi_{k*}(x))\right) &= \partial_x \left( \phi_k(x) \partial_x \phi^*_{k}(x) - \phi^*_k(x) \partial_x \phi_{k}(x) \right)\\ &= \left( \phi_k(x) \partial^2_x \phi^*_{k}(x) - \phi^*_k(x) \partial^2_x \phi_{k}(x) \right) \\
			&= \left( \phi_k(x) \left(u(x) + (k^*)^2 \right) \phi^*_{k}(x) - \phi^*_k(x) \left(u(x) + (k)^2 \right) \phi_{k}(x) \right) \\
			&= \left(k^2 - (k^*) \right) \phi_k(x) \phi^*_k(x).
		\end{align*}
		Now by considering the asymptotic behavior presented in Equation~\eqref{eqn:asympa(k)zeros} we conclude
		\begin{equation}
			\int_{-\infty}^{+\infty}\left(W(\phi_k(x), \phi_{k^*}(x))\right) dx \to 0.
		\end{equation}
		This implies
		\begin{equation}
			\left((k^*)^2 - k^2\right) \int_{-\infty}^{\infty} \lvert \psi_k(x) \rvert^2 dx = 0
		\end{equation}
		which leads to
		\begin{equation}
			0 = \left((k^*)^2 - k^2\right) = 4 i \delta_0 \xi_0 \quad \Rightarrow \quad \xi_0 = 0.
		\end{equation}
		Thus, since $\delta_0>0$ and $\xi_0=0$ the zeros of $B(k)$ lie on the positive imaginary axis.
		
		{\bf Step 2:} Let us now prove that the zeros of $B(k)$ are simple. Consider once again the Schr\"odinger equation~\eqref{eqn:Schorosimple} for $\phi_k(x)$. Now multiply it by $\partial_k\phi_k(x)$ and obtain
		\begin{equation}
			\label{eqn:diffB}
			\partial_k \phi_k(x) \partial_x^2 \phi_k(x) + \partial_k \phi_k(x) \left(u(x) + k^2 \right) \phi_k(x) =0.
		\end{equation}
		In a similar way, by computing the derivative of Equation~\eqref{eqn:Schorosimple} w.r.t.\ $k$ and multiplying the result with $\phi_k(x)$ we get
		\begin{equation}
			\label{eqn:diffA}
			\phi_k(x) \partial_k \partial_x^2 \phi_k(x) + \left(u(x) + k^2 \right) \phi_k(x) \partial_k \phi_k(x) + 2 k \left(\phi_k(x)\right)^2 = 0.
		\end{equation}
		The difference between Equations~\eqref{eqn:diffB} and~\eqref{eqn:diffA} gives
		\begin{equation}
			\bigg(\phi_k(x) \partial_k \partial_x \phi_k(x) - \partial_x \phi_k(x) \partial_k \phi_k(x)  \bigg) + 2 k\, \bigg(\phi_k(x)\bigg)^2 = 0.
		\end{equation} 
		This implies
		\begin{equation}
			\lim_{x\to +\infty}W(\phi_k(x), \partial_k \phi_k(x)) - \lim_{x\to -\infty}W(\phi_k(x), \partial_k \phi_k(x)) = - 2 k \int_{-\infty}^{\infty} \phi^2_k(x) dx. \label{eqn:integral}
		\end{equation}
		Now compute
		\begin{equation}
		\label{eqn:Wronskianproof}
		\begin{aligned}
			& W \left(\partial_k \phi_k(x), \psi_k(x)\right) + W \left(\phi_k(x), \partial_k \psi_k(x)\right)  \\
			& \quad = \partial_k \phi_k(x) \partial_x \psi_k(x) - \partial_k \partial_x \phi_k(x) \psi_k(x) + \phi_k(x) \partial_x \partial_k \psi_k(x) - \partial_x \phi_k(x) \partial_k \psi_k(x) \\
			& \quad = \partial_k W(\phi_k (x), \psi_k(x)) = \, \partial_k (2 i k B(k) ) = 2 i B(k) + 2 i k\, \partial_k B(k).
		\end{aligned}
		\end{equation}
		If $B(k)$ has a zero in $i \sigma$ then we obtain
		\begin{equation}
			\phi_{i \sigma}(x) = B(i \sigma) \psi_{i \sigma}(x)
		\end{equation}
		and Equation~\eqref{eqn:Wronskianproof} becomes
		\begin{equation}
			W \left(\partial_k \phi_k(x), \psi_k(x)\right) + W \left(\phi_k(x), \partial_k \psi_k(x)\right) = 2 i A(i \sigma) - 2 \sigma \partial_{k} A( i \sigma).
		\end{equation}
		Now compute
		\begin{equation}
			A(k)^2 W \left(\psi_k(x), \partial_k \psi_k(x) \right) - W(\phi_k(x),\partial_k \phi_k(x)) = - 2 B(i \sigma) \sigma\,  \partial_k A(i\sigma)
		\end{equation}
		which, when taking the limit for $x \to - \infty$, and comparing it with equation~\eqref{eqn:integral}:
		\begin{equation}
			- 2 i \sigma \int_{-\infty}^{\infty} \phi_{i \sigma}(x)^2 dx = - 2 \sigma B(i \sigma)\, \partial_k A(i \sigma).
		\end{equation}
		But using Equations~\eqref{eqn:asymptotic} we obtain for $A(i \sigma) = 0$
		\begin{equation}
			\partial_k A(i\sigma) = - \frac{i}{B(i \sigma)} \int_{-\infty}^{\infty} \phi^2_{i \sigma}(x)\, dx. \label{eqn:simpleconclusion}
		\end{equation}
		Moreover, since $\phi_{i \sigma}(x)$ is a real function, Equation~\eqref{eqn:simpleconclusion} is not vanishing and then the zeros are simple.

		
		{\bf Step 3:} The last step is to prove that the zeros of $B(k)$ are finite. Let us recall that by Proposition~\ref{lemma:221}
		\begin{equation}
			\lim_{\lvert k \rvert \to \infty} B(k) \to 1.
		\end{equation}
		Moreover $B(k)$ is continuous for positive imaginary part of $k$, and therefore it possess a finite number of zeros. For details see~\cite{Segur_1973}.
		
		Consequently since $M_k(x)/B(k)$ is the fraction of two holomorphic functions, it is meromorphic with a finite number $N$ of isolated poles on the positive imaginary axis. This permits to write it as a sum of an holomorphic function $\eta_k(x)$ and a sum of $N$ terms that describes the behaviour of the poles:
		\begin{equation}
			\frac{M_k(x)}{B(k)} := \eta_k(x) + \sum_{j=1}^{N} \frac{\alpha_j(x)}{k - i k_j}, \label{eqn:splitholo}
		\end{equation} 
		where $\alpha_j(x)$ are real functions that describes the behavior around the poles of $B(k)$.
		Since
		\begin{equation}
			\frac{M_k(x)}{B(k)} = \overline{N}_k(x) + \rho(k) e^{2 i k x} \overline{N}(x,-k),
		\end{equation}
		and by integrating around of each pole of Equation~\eqref{eqn:splitholo} and by using the residue theorem we get
		\begin{equation}
			\alpha_j(x) = C_j \overline{N}_{-i k_j}(x) e^{-2 k_j x}
		\end{equation}
		where $C_j$ are normalization constants. Finally we compute
		\begin{equation}
			\eta_k(x) = \overline{N}_k(x) + \rho(k) e^{2i k x} \overline{N}_{-k}(x) - \sum_{j=1}^N \frac{C_j}{k - i k_j} e^{-2 k_j x} \overline{N}_{-ik_j} (x).
		\end{equation}
	\end{proof}
	
	The two propositions above permit to define

	\begin{framed}
	\begin{definition}[Scattering data]
		We call the set
		\begin{equation}
			S = \left\{\{k_j, C_j\}_{1\leq j \leq N}, \rho(k), B(k) \right\}
		\end{equation}
		the {\em scattering data} where
		\begin{itemize}
			\item $\{k_j, C_j\}_{1\leq j \leq N}$ consist of the poles and the associated normalized coefficients,
			\item $\rho(k)$ are the transmission coefficients,
			\item $B(k)$ is the incoming wave coefficient.
		\end{itemize}
	\end{definition}
	\end{framed}

	We clarify that, for now, these scattering data and the set $S$ are time-independent because we are studying a time-independent setting. However, from the next section onward, the scattering data will incorporate time dependence.
	

\clearpage

\section{Time evolution of scattering data}

Let us consider again the Lax pair, and in particular the time-evolution operator $M$. As in the previous subsection, the asymptotic behavior will take a central role in the procedure. Let us consider the transformation
\begin{equation}
	v_k(t,x) = M_k(t,x) e^{-i k x}.
\end{equation}
Then the equation $\partial_t v = M v$ becomes
\begin{equation}
	\partial_t M_k(t,x) = \left( \gamma - 4ik^3 + u_x + 2 i k u(t,x) \right) M_k(t,x) + \left(4 k^2 - 2 u(t,x)\right) \partial_x M_k(t,x).
\end{equation}
The asymptotic behaviour of $M_k(x,t)$ is again given by
\begin{equation}
	\begin{cases}
		M_k(t,x) \sim 1 \quad \text{for} \, x\to - \infty , \\
		M_k(t,x) \sim B(k,t) + A(k,t) e^{2 i k x} \quad \text{for} \, x\to - \infty
	\end{cases}
\end{equation}
since the asymptotic analysis of Equation~\eqref{eqn:asymptMN} is valid for each value of $t$. This permits to consider for $x\to -\infty$ the equation
\begin{equation}
	\gamma - 4ik^3 = 0,
\end{equation}
by using also the Fadeev condition given by $u \in P_1$. For $x\to+\infty$ we obtain
\begin{equation}
	\label{eqn:timeevolution(ab)}
	\partial_t B(k,t)  + \partial_t A(k,t) e^{2 i k x} = 8 i k^3 A(k,t) e^{2 i k x} \quad \mbox{and thus} \quad  \begin{cases}
		B(k,t) = B(k,0), \\
		A(k,t) = A(k,0) e^{8 i k^3 t}.
	\end{cases}
\end{equation}
Since $B(k,t) = B(k,0)$ for all $t$, the locations of the zeros are fixed during the time evolution.
From Equation~\eqref{eqn:timeevolution(ab)} we recover the transmission coefficient by using Equation~\eqref{eqn:defrefandtrans}:
\begin{equation}
	\tau(k,t) = \tau(k,0) \qquad \mbox{and} \qquad \rho(k,t) = \rho(k,0) e^{ 8 i k^3 t}.
\end{equation}
This is enough to time evolve the scattering data set $S$ and we obtain
\begin{align}
\label{eqn:timedepscattering}
	S = \bigg\{ \bigg\{&k_j(t) = k_j, C_j(t) = C_j(0) e^{ 8 i k_j^3 t} \bigg\}_{1 \leq j \leq N} , \ \
	\rho(k,t) = \rho(k,0) e^{ 8 i k^3 t}, \ \  B(k,t) = B(k,0) \bigg\} .
\end{align}


\clearpage

\section{Inverse scattering}
	The aim of this section is to describe the final step of the procedure illustrated in Figure~\ref{fig:scatteringinv}: the construction of the time-dependent solution $ u(t,x) $ from the time-dependent scattering data. To simplify the notation in this section, we will omit explicit references to the time dependence of objects. Specifically, we will describe how to transform the general scattering data (whether time-independent or time-dependent) into the corresponding solution.
	
	Moreover, we will assume that the function $ B(k) $ does not have any poles, allowing us to use a simpler version of the equation, namely equation~\eqref{eqn:moverA}.
	Even with this restrictive assumption, it is possible to extend the procedure to cases where $ B(k) $ has poles, see \cite{ablowitzClarkson}.
	
	Let us consider again the Volterra integral equation
	\begin{equation*}
		\overline{N}_k(x) = 1 - \frac{1}{2 i k}\int_{x}^{+\infty} \left(1 - e^{2 i k (x-y)} \right) u(y) \overline{N}_k(y) d y.
	\end{equation*}
 	Since for $\Im(k)<0$ we have
 	\begin{equation*}
 		\lim_{\lvert k \rvert \to +\infty } e^{2 i k (x-y)} \to 0
 	\end{equation*}
 	we obtain
 	\begin{equation}
 	\label{eqn:lineN}
 		\overline{N}_k(x) \underset{\lvert k \rvert \to+\infty}{\sim} 1 - \frac{1}{2 i k} \int_{x}^{\infty} u(y) dy.
 	\end{equation}
 	On the other side, by using the Cauchy representation theorem for an analytic function $\mathcal{F}(y)$ with $ 		\lim_{\lvert y \rvert \to+\infty} \mathcal{F}(y) \to 0$ we obtain
 	\begin{equation}
 		\mathcal{F}(k) = \frac{1}{2 \pi i} \oint \frac{\mathcal{F}(\xi)}{\xi - k} d\xi.
 	\end{equation}
 	To compute this integral we can chose two different paths: one that passes above the singularity $\xi = k$ called $\gamma^+$ or the one below $\gamma_-$. To evaluate this integral, we apply Jordan lemma.
 	\begin{lemma}
 		Suppose $f(z)$ is an analytic function at all points in the upper (lower) half plane, that are exterior to a circle of radius $R_0$ centered in the origin. Let $C_R$ the upper semicircle with radius $R>R_0$ and that for all point $z \in C_R$ there exist a positive constant $\Theta_R$ such that:
 		\begin{equation}
 			\begin{dcases}
 				\lvert f(z) \rvert \leq \Theta_R, \\
 				\lim_{R \to +\infty} M_R = 0.
 			\end{dcases}
 		\end{equation}
 		Then for every constant $a \in \mathbb{R}^+$, then:
 		\begin{equation}
 			\lim_{R \to +\infty} \int_{C_R} f(z) e^{\imath a z} dz = 0,
 		\end{equation}
 	\end{lemma}
 	\begin{proof}
 		See~\cite{Brown_2014}.
 	\end{proof}
	We can define the two projections:
 	\begin{equation}
 		\begin{dcases}
 			\ \left(\Pi^+ \mathcal{F}\right)(\kappa) = \frac{1}{2 \pi i} \int_{-\infty}^{\infty} \frac{\mathcal{F}(\xi)}{\xi - k + (i 0)} d\xi,  \\
 			\ \left(\Pi^- \mathcal{F}\right)(\kappa) = \frac{1}{2 \pi i} \int_{-\infty}^{\infty} \frac{\mathcal{F}(\xi)}{\xi - k - (i 0)} d\xi.
 		\end{dcases}
 	\end{equation}
 	Now rewrite Equation~\eqref{eqn:moverA} to
 	\begin{equation}
 		{M_k(x)}{B(k)} - \overline{N}_k(x) - \rho(k) e^{2 i k x} \overline{N}_{-k}(x) = 0,
 	\end{equation}
 	and then compute its projection integral using $\gamma_-$:
 	\begin{align}
 		0 &= \frac{1}{2 i \pi} \int_{-\infty}^{\infty}\frac{1}{k - \xi} \left({M_k(x)}{B(k)} - \overline{N}_k(x) - \rho(k) e^{2 i k x} \overline{N}_{-k}(x)\right) \\
 		&= -1 + \overline{N}_k(x) - \frac{1}{2 i \pi} \int_{-\infty}^{\infty}\frac{1}{k - \xi} \left( \rho(k) e^{2 i k x} \overline{N}_{-k}(x)\right)
 	\end{align}
 	where we have used Propositions~\ref{lemma:221} and~\ref{prop:secondprop}, and in particular
 	\begin{equation*}
 		\lim_{\lvert k \rvert \to \infty} \frac{M_k(x)}{B(k)} \to 1,
 	\end{equation*}
 	and the analytic properties of $\overline{N}_k(x)$ for $\Im(k)<0$, and of $M_k(x)$ and $a(k)$ for $\Im(k)>0$. Taking the limit $\lvert k \rvert \to \infty$ and using Equation~\eqref{eqnNforNbar} we get
 	\begin{equation}
 		\overline{N}_k(x) = 1 - \frac{1}{2 k \pi i} \int_{-\infty}^{+\infty} \rho(\xi) N_{\xi}(x) d\xi. \label{eqn:Ninvscattering}
 	\end{equation}
 	Comparing the above equation with equation~\eqref{eqn:lineN}, we obtain
 	\begin{equation*}
 		1 - \frac{1}{2 k \pi i} \int_{-\infty}^{+\infty} \rho(\xi) N_{\xi}(x) d\xi = 1 - \frac{1}{2 i k} \int_{x}^{\infty} u(y) dy.
 	\end{equation*}
 	Taking the $x$-derivative of both sides the result is
 	\begin{equation*}
 		u(x) = -\partial_x \left( \frac{1}{\pi} \int_{-\infty}^{\infty} \rho(\xi) N_{\xi}(y) d\xi\right).
 	\end{equation*}

 	Now a final note regarding the case for which $B(k)$ has zeros (for details see~\cite{ablowitzClarkson}):

 	\begin{framed}
 	\begin{remark}
 	 The two integral representations for $\overline{N}_k(x)$ have to consider the poles of $B(k)$, since according to the residue theorem (see for instance \cite{rudin}) the integral over a path $\partial \Omega$ is of the form
 	\begin{equation}
 		\frac{1}{2 i \pi} \int_{\partial \Omega} \mathcal{F}(z) dz = \sum_{l} \text{Res}_{\mathcal{F}}(z_l),
 	\end{equation} 
 	where $z_l$ are the poles of the function $\mathcal{F}$.
 	Thus if the zeros of $B(k)$ are located in $k_l=i \sigma_l$ as in the Proposition~\ref{prop:secondprop} then Equation~\eqref{eqn:Ninvscattering} becomes
 	\begin{equation*}
 		\overline{N}_k(x) = 1 - \left( \sum_{l} \frac{C_l e^{- 2 \sigma_l x} \overline{N}_{-i \sigma_l}(x)}{k + i \sigma_l}\right) + \frac{1}{2 \pi i k} \int_{-\infty}^{+\infty} \rho(\xi) N_{\xi}(x) d\xi
 	\end{equation*}
 	and eventually
 	\begin{equation*}
 		u(x) = \partial_x \left(2 i \sum_{j=1}^{N} C_j N_j(x) - \frac{1}{\pi} \int_{-\infty}^{\infty} \rho(\xi) N_{\xi}(x) d \xi \right).
 	\end{equation*}
 \end{remark}
 \end{framed}




 



\backmatter

\printindex


\begin{thebibliography}{30mm}


\bibitem[Ablowitz $\&$ Clarkson]{ablowitzClarkson} {\sc Ablowitz, M. J.; Clarkson, P. A.:} {\it Solitons, nonlinear evolution equations and inverse scattering.} London Mathematical Society Lecture Note Series, 149. Cambridge University Press, Cambridge, 1991. xii+516 pp.

\bibitem[Ablowitz $\&$ Kaup $\&$ Newell $\&$ Segur]{AKNS} {\sc Ablowitz, Mark J.; Kaup, David J.; Newell, Alan C.; Segur, Harvey:} {\it The inverse scattering transform-Fourier analysis for nonlinear problems.} Studies in Appl. Math. 53 (1974), no. 4, 249--315.

\bibitem[Abraham $\&$ Marsden] {abrahamMarsden} {\sc Abraham, Ralph; Marsden, Jerrold E.:} {\it Foundations of mechanics.} Second edition, revised and enlarged. With the assistance of Tudor Ratiu and Richard Cushman. Benjamin/Cummings Publishing Co., Inc., Advanced Book Program, Reading, Mass., 1978. xxii+m-xvi+806 pp. 

\bibitem[Aktosun]{aktosun} {\sc Aktosun, Tuncay:} {\it Inverse scattering transform, KdV, and solitons.} Current trends in operator theory and its applications, 1--22, Oper. Theory Adv. Appl., 149, Birkhäuser, Basel, 2004. 










\bibitem[Audin 1991] {audin-torus} {\sc Audin, Mich\`ele:} {\it The topology of torus actions on symplectic manifolds.} Translated from the French by the author. Progress in Mathematics, 93. Birkh\"auser Verlag, Basel, 1991. 181 pp. 


\bibitem[Audin 2008] {audin-integrability} {\sc Audin, Mich\`ele} {\it Hamiltonian systems and their integrability}. Translated from the 2001 French original by Anna Pierrehumbert. Translation edited by Donald Babbitt. SMF/AMS Texts and Monographs, 15. American Mathematical Society, Providence, RI; Soci\'et\'e Math\'ematique de France, Paris, 2008. xii+149 pp. 

\bibitem[Babelon $\&$ Bernard $\&$ Talon] {babelonBernardTalon} {\sc Babelon, Olivier; Bernard, Denis; Talon, Michel:} {\it Introduction to classical integrable systems.} Cambridge Monographs on Mathematical Physics. Cambridge University Press, Cambridge, 2003. xii+602 pp.


\bibitem[Batlle] {batlle} {\sc Batlle, Carles:} {\it Lecture notes on KdV hierarchies and pseudodifferential operators.} 

{\tiny \verb|https://mat-web.upc.edu/people/carles.batlle/fitxers/kdv.pdf|}

\bibitem[Bolsinov $\&$ Fomenko] {bolsinov-fomenko} {\sc Bolsinov, A.\ V.; Fomenko, A.\ T.:} {\it Integrable Hamiltonian systems. Geometry, topology, classification.} Translated from the 1999 Russian original. Chapman $\&$ Hall/CRC, Boca Raton, FL, 2004. xvi+730 pp.


\bibitem[Boussinesq]{boussinesq} {\sc Boussinesq, J.:} {\it Essai sur la th\'eorie des eaux courantes.} M\'emoires pr\'esent\'es par divers savants \`a l’Acad. des Sci. Inst. Nat. France, XXIII (1877), pp.\ 1--680.

\bibitem[Brilleslijper $\&$ Fabert]{BrilleslijperFabert}
{\sc Brilleslijper, R.; Fabert, O.:} {\it From Euclidean field theory to hyperk\"ahler Floer theory via regularized polysymplectic geometry.} 20p., arXiv:2311.18485

\bibitem[Bro]{Brown_2014}
{\sc Brown, James Ward and Churchill, Ruel Vance}: {\em Complex Variables and Applications}. McGraw-Hill Education, Ninth Edition, New York, NY, 2014, ISBN: 978-0-07-338317-0, 461 pages.


\bibitem[Cannas da Silva] {cannas-da-silva} {\sc Cannas da Silva, Ana:} {\it Lectures on symplectic geometry.} Lecture Notes in Mathematics, 1764. Springer-Verlag, Berlin, 2nd edition, 2008. xii+217 pp.





\bibitem[de Jager] {deJager} {\sc de Jager, E.M.:} {\it On the origin of the Korteweg-de Vries equation.} A. Festkolloquium `Rudolf Gorenflo. Fluide aus fraktionaler Sicht'. [Honorary colloquium `Rudolf Gorenflo. Fluids from a fractional viewpoint'] B. Hans Gebeleins Turbulenz aus stochastischer Sicht, Wellen von Korteweg und de Vries, zellul\"are Diffusion u.a.\ 171--195 (2011), Berlin: Berliner Mathematische Gesellschaft, Volume 19. 

{\small \verb|https://arxiv.org/abs/math/0602661|}


\bibitem[Dickey] {dickey} {\sc Dickey, L. A.:} {\it Soliton equations and Hamiltonian systems.} Second edition. Advanced Series in Mathematical Physics, 26. World Scientific Publishing Co., Inc., River Edge, NJ, 2003. xii+408 pp. 

\bibitem[Drazin $\&$ Johnson]{drazinJohnson} {\sc Drazin, P. G.; Johnson, R. S.:} {\it Solitons: an introduction.} Cambridge Texts in Applied Mathematics. Cambridge University Press, Cambridge, 1989. {\rm xii}+226 pp.





\bibitem[Dunajski]{dunajski} {\sc Dunajski, Maciej:} {\it Lecture notes on integrable systems.} See his webpage
{\small \verb|http://www.damtp.cam.ac.uk/user/md327/|}

\bibitem[Eckhaus $\&$ van Harten]{eckhausVanHarten} {\sc Eckhaus, Wiktor; van Harten, Aart:} {\sc The inverse scattering transformation and the theory of solitons. An introduction.} North-Holland Mathematics Studies, 50. North-Holland Publishing Co., Amsterdam-New York, 1981. xi+222 pp. 



\bibitem[Farlow]{farlow} {\sc Farlow, Stanley J.:} {\it Partial differential equations for scientists and engineers.} Revised reprint of the 1982 original. Dover Publications, Inc., New York, 1993. xii+414 pp.





\bibitem[Gardner]{gardner}{\sc Gardner, Clifford S.:} {\it Korteweg-de Vries equation and generalizations. IV. The Korteweg-de Vries equation as a Hamiltonian system.} J. Mathematical Phys. 12 1971 1548--1551.

\bibitem[Gardner $\&$ Greene $\&$ Kruskal $\&$ Miura]{gardnerGreeneKruskalMiura} {\sc Gardner, Clifford S.; Greene, John M.; Kruskal, Martin D.; Miura, Robert M.:} {\it Korteweg-deVries equation and generalization. VI. Methods for exact solution.} Comm. Pure Appl. Math. 27 (1974), 97--133. 



\bibitem[Hall]{qmhall} {\sc Hall, Brian C.} {\it Quantum Theory for Mathematicians} Graduate Texts in Mathematics, Springer New York, ISBN: 9781461471165

\bibitem[Hofer $\&$ Zehnder] {hofer-zehnder} {\sc Hofer, Helmut; Zehnder, Eduard:} {\it Symplectic invariants and Hamiltonian dynamics}. Reprint of the 1994 edition. Modern Birkh\"auser Classics. Birkh\"auser Verlag, Basel, 2011. xiv+341 pp.

\bibitem[Hohloch]{hohloch} {\sc Hohloch, Sonja:} {\it Hyperk\"ahler Floer theory as infinite dimensional Hamiltonian system.} Proc. Amer. Math. Soc. 143 (2015), no. 8, 3519--3524.

\bibitem[Hohloch $\&$ Noetzel $\&$ Salamon]{HNS} {\sc Hohloch, Sonja; Noetzel, Gregor; Salamon, Dietmar A.:} {\it  Hypercontact structures and Floer homology.} Geom. Topol. 13 (2009), no. 5, 2543--2617.








 











\bibitem[Korteweg $\&$ de Vries] {kortewegDeVries} {\sc Korteweg, D. J.; de Vries, G.:} {\it On the change of form of long waves advancing in a rectangular canal, and on a new type of long stationary waves.} Philos. Mag. (5) 39 (1895), no. 240, 422--443. 

\bibitem[Kruskal $\&$ Miura $\&$ Gardner $\&$ Zabusky] {kruskalMiuraGardnerZabusky} {\sc Kruskal, Martin D.; Miura, Robert M.; Gardner, Clifford S.; Zabusky, Norman J.:} {\it Korteweg-de Vries equation and generalizations. V. Uniqueness and nonexistence of polynomial conservation laws.} J. Mathematical Phys. 11 1970 952--960.

\bibitem[Kuksin]{kuksin} {\sc Kuksin, Sergei B.:} {\it Analysis of Hamiltonian PDEs.} Oxford Lecture Series in Mathematics and its Applications, 19. Oxford University Press, Oxford, 2000. xii+212 pp.
 
\bibitem[Kupershmidt] {kupershmidt} {\sc Kupershmidt, Boris A.:} {\sc KP or mKP. Noncommutative mathematics of Lagrangian, Hamiltonian, and integrable systems.} Mathematical Surveys and Monographs, 78. American Mathematical Society, Providence, RI, 2000. xx+600 pp.  
 
 
 
 
 
 




\bibitem[McDuff $\&$ Salamon] {mcduffSalamon} {\sc Introduction to symplectic topology.} Third edition. Oxford Graduate Texts in Mathematics. Oxford University Press, Oxford, 2017. xi+623 pp. 
 



 

 

\bibitem[Miura] {miura} {\sc Miura, Robert M.:} {\it Korteweg-de Vries equation and generalizations. I. A remarkable explicit nonlinear transformation.} J. Mathematical Phys. 9 1968 1202--1204.

\bibitem[Miura $\&$ Gardner $\&$ Kruskal]{miuraGardnerKruskal} {\sc Miura, Robert M.; Gardner, Clifford S.; Kruskal, Martin D.:} {\it Korteweg-de Vries equation and generalizations. II. Existence of conservation laws and constants of motion.} J. Mathematical Phys. 9 1968 1204--1209. 

\bibitem[Miwa $\&$ Jimbo $\&$ Date] {miwaJimboDate} {\sc Miwa, T.; Jimbo, M.; Date, E.:} {\it Solitons. Differential equations, symmetries and infinite-dimensional algebras.} Translated from the 1993 Japanese original by Miles Reid. Cambridge Tracts in Mathematics, 135. Cambridge University Press, Cambridge, 2000. x+108 pp.

\bibitem[Olver]{olver}{\sc Olver, Peter J.:} {\it Applications of Lie groups to differential equations.} Second edition. Graduate Texts in Mathematics, 107. Springer-Verlag, New York, 1993. xxviii+513 pp.









\bibitem[Reed $\&$ Simons] {reedsimons2} {\sc Michael Reed; Barry Simons} {\it Mathematical Physics: Fourier Analysis and self-adjoint operators} ISBN: 9780123958266

\bibitem[Rudin] {rudin} {\sc Rudin, W.} {\it Real and Complex Analysis}, McGraw-Hill, ISBN: 9780071002769

\bibitem[Sakurai $\&$ Napolitano] {qmsakurai} {\sc Sakurai, J. J. and Napolitano, Jim} {\it Modern Quantum Mechanics} Cambridge University Press, ISBN: 9781108473224

\bibitem[Schneider $\&$ Uecker] {schneiderUecker} {\sc Schneider, Guido; Uecker, Hannes:} {\sc Nonlinear PDEs. A dynamical systems approach.} Graduate Studies in Mathematics, 182. American Mathematical Society, Providence, RI, 2017. xiii+575 pp.

\bibitem[Scott Russell] {scottRussell} {\sc Scott Russell, John:} {\it Report on Waves} (made to the meetings in 1843 and 1844). Report of the fourteenth meeting of the British Association for the Advancement of Science, held at York in September 1844, London, John Murray, Albemarle street, 1845.

\bibitem[Seg]{Segur_1973}
{\sc Segur, Harvey}: {\em The Korteweg-de Vries equation and water waves. Solutions of the equation. Part 1}. J. Fluid Mech., Vol. 59, No. 4, pp. 721--736, 2018, DOI: 10.1017/S0022112073001813.


\bibitem[Su $\&$ Gardner]{suGardner} {\sc Su, C. H.; Gardner, C. S.:} {\it Korteweg-de Vries equation and generalizations. III. Derivation of the Korteweg-de Vries equation and Burgers equation.} J. Mathematical Phys. 10 1969 536--539.



\bibitem[Vilasi] {vilasi} {\sc Vilasi, Gaetano:} {\it Hamiltonian dynamics. World Scientific Publishing Co., Inc., River Edge, NJ, 2001. xvi+440 pp.}




\bibitem[\vungoc\ 2006] {vungoc book} {\sc \vungoc, San:} {\it Syst\`emes int\'egrables semi-classiques: du local au global. (French) [Semiclassical integrable systems: from the local to the global]}. Panoramas et Synth\`eses [Panoramas and Syntheses], 22. Soci\'et\'e Math\'ematique de France, Paris, 2006. vi+156 pp. 










\bibitem[Zhidkov] {zhidkov} {\sc Zhidkov, Peter E.:} {\it Korteweg-de Vries and nonlinear Schr\"odinger equations: qualitative theory.} Lecture Notes in Mathematics, 1756. Springer-Verlag, Berlin, 2001. vi+147 pp.

\bibitem[KdVFaddev] {kdvfaddev} {\sc V. E. Zakharov; L. D. Faddeev} {\it Korteweg-de Vries equation: a completely integrable Hamiltonian system} Journal of Functional Analysis and Its Applications, 1971.

\end{thebibliography}
\end{document}